\documentclass[prl, reprint,showpacs]{revtex4-1}

\usepackage{ fullpage, amssymb, amsmath, dsfont, graphicx, amsthm}
\usepackage{epstopdf} 

\input{Qcircuit}

\renewcommand{\th}{^{\mathrm{th}}}
\newcommand{\id}{\mathds{1}}
\newcommand{\eq}[1]{(\ref{#1})}
\newcommand{\re}{\mathrm{Re}}
\newcommand{\braket}[2]{\langle #1|#2\rangle}
\newtheorem{theorem}{Theorem}

\newtheorem{definition}{Definition}
\newtheorem*{cor}{Corollary}

\begin{document}

\title{Grover search and the no-signaling principle}
\author{Ning Bao}
\affiliation{Institute for Quantum Information and Matter}
\affiliation{Walter Burke Institute for Theoretical Physics, California Institute of Technology 452-48, Pasadena, CA 91125, USA }
\author{Adam Bouland}
\affiliation{Computer Science and Artificial Intelligence Laboratory, Massachusetts Institute of Technology, Cambridge, MA 02139 USA}
\author{Stephen P. Jordan}
\affiliation{National Institute of Standards and Technology,
  Gaithersburg, MD, 20899}
\affiliation{Joint Center for Quantum Information and Computer Science, University of Maryland, College Park, MD 20742 USA}


\bibliographystyle{naturemag}

\begin{abstract}
Two of the key
properties of quantum physics are the no-signaling principle and the
Grover search lower bound. That is, despite admitting
stronger-than-classical correlations, quantum mechanics does not imply
superluminal signaling, and despite a form of exponential parallelism,
quantum mechanics does not imply polynomial-time brute force solution
of \textsf{NP}-complete problems. Here, we investigate the degree to which
these two properties are connected. We examine four classes of
deviations from quantum mechanics, for which we draw inspiration from
the literature on the black hole information paradox. We show that in these models, the
physical resources required to send a superluminal signal scale
polynomially with the resources needed to speed up Grover's
algorithm. 
Hence the
no-signaling principle is equivalent to the inability to solve \textsf{NP}-hard problems efficiently by brute force
within the classes of theories analyzed.

\end{abstract}

\pacs{03.67.-a, 03.65.-w, 04.70.Dy}
\maketitle

\section{Introduction}
Recently the firewalls paradox \cite{AMPS,Braunstein13} has shown that
our understanding of quantum mechanics and general relativity appear to be inconsistent at the event horizon of a black hole. Many of the leading
proposals to resolve the paradox involve modifying quantum
mechanics. For example, the final-state projection model of Horowitz
and Maldecena \cite{Horowitz_Maldecena} and the state dependence model
of Papadodimas and Raju \cite{PR14} are modifications to quantum
theory which might resolve the inconsistency.

One reason to be skeptical of such modifications of quantum mechanics
is that they can often give rise to superluminal signals, and hence
introduce acausality into the model. For example, Weinberg nonlinearities allow for superluminal signaling \cite{Gisin,Polchinski91}. This is generally seen as
unphysical. 
In contrast, in standard quantum theory, entanglement does not give rise to superluminal signaling.

Another startling feature of such models is that they might allow one to construct computers far more powerful even than conventional quantum computers.
In particular, they may allow one to solve \textsf{NP}-hard problems in polynomial time. 
 \textsf{NP}-hard problems refer to those problems for which the solution can be \emph{verified} in polynomial time, but for which there are exponentially many possible solutions.  
It is impossible for standard quantum computers to solve \textsf{NP}-hard problems efficiently by searching over all possible solutions. This is a consequence of the query complexity lower bound of Bennett, Bernstein, Brassard and Vazirani  \cite{BBBV},  which shows one cannot search an unstructured list of $2^n$ items in fewer than $2^{n/2}$ queries with a quantum computer. (Here a \emph{query} is an application of a function $f$ whose output indicates if you have found a solution. The \emph{query complexity} of search is the minimum number of queries to $f$, possibly in superposition, required to find a solution.) 
This bound is achieved by Grover's search algorithm \cite{Grover}. 
In contrast, many modifications of quantum theory allow quantum computers to search an exponentially large solution space in polynomial time. 
For example, quantum computers equipped with postselection  \cite{Aaronson_postBQP}, Deutschian closed timelike curves \cite{Brun, AaronsonCTC,
BennettCTC},  or nonlinearities \cite{Abrams_Lloyd, MW14, MW15, MW13, CY15}
all admit poly-time solution of \textsf{NP}-hard problems by brute force search. 

In this paper we explore the degree to which superluminal
signaling and speedups over Grover's algorithm are connected. We consider several
modifications of quantum mechanics which are inspired by resolutions
of the firewalls paradox. For each modification, we show that the
theory admits superluminal signaling if and only if it admits a query
complexity speedup over Grover search.
Furthermore, we establish a \emph{quantitative} relationship between superluminal signaling and speedups over Grover's algorithm.
More precisely, we show that if one can transmit one classical
bit of information superluminally using $n$ qubits and $m$ operations,
then one can speed up Grover search on a system of poly$(n,m)$ qubits
with poly$(n,m)$ operations, and vice versa. In other words, the
ability to send a superluminal signal with a reasonable amount of
physical resources is equivalent to the ability to violate the Grover
lower bound with a reasonable amount of physical resources. 
Therefore the
no-signaling principle is equivalent to the inability to solve \textsf{NP}-hard problems efficiently by brute force
within the classes of theories analyzed. 

Note that in the presence of nonlinear dynamics, density matrices are no longer equivalent to ensembles of pure states. Here, we consider measurements to produce probabilistic ensembles of post-measurement pure states and compute the dynamics of each of these pure states separately. Alternative formulations, in particular Everettian treatment of measurements as entangling unitaries, lead in some cases to different conclusions about superluminal signaling. See e.g. \cite{Deutsch}.

\section{Results}
We consider four modifications of quantum mechanics,
which are inspired by resolutions of the firewalls paradox. The first
two are ``continuous'' modifications in the sense that they have a
tunable parameter $\delta$ which quantifies the deviation from quantum
mechanics. The second two are ``discrete" modifications in which standard quantum mechanics is supplemented by one additional operation.

\subsection{Final state projection}

The first ``continuous" modification of quantum theory we consider is the final state projection model of
Horowitz and Maldecena \cite{Horowitz_Maldecena}, in which the black
hole singularity projects the wavefunction onto a specific quantum
state. This can be thought of as a projective measurement with
postselection, which induces a linear (but not necessarily unitary)
map on the projective Hilbert space. (In some cases it is possible for the
Horowitz-Maldecena final state projection model to induce a perfectly
unitary process $S$ for the black hole, but in general interactions between the collapsing body
and infalling Hawking radiation inside the event horizon 
induce deviations from unitarity \cite{Gottesman_Preskill}.)  Such linear but non-unitary maps allow both superluminal signaling and speedups over Grover
search.  Any non-unitary map $M$ of condition number
$1+\delta$ allows for superluminal signaling with channel capacity
$O(\delta^2)$ with a single application of $M$. The protocol for signaling is simple - suppose Alice has the ability to apply $M$, and suppose Alice and Bob share the entangled state
\begin{equation}\frac{1}{\sqrt{2}} \left( \ket{\phi_0} \ket{0} + \ket{\phi_1} \ket{1} \right).\end{equation}
where $\ket{\phi_0}$ and $\ket{\phi_1}$ are the minimum/maximum singular vectors of $M$, respectively.  If Alice chooses to apply $M$ or not, then Bob will see a change in his half of the state, which allows signaling with channel capacity $\sim \delta^2$. 
Furthermore, it is also possible for Bob to signal superluminally to Alice with the same state - if Bob chooses to measure or not to measure his half of the state, it will also affect the state of Alice's system after Alice applies $M$. 
So this signaling is bidirectional, even if only one party has access to the
non-unitary map. 
In the context of the black hole information paradox, this implies the acausality in the final state projection model could be present even far away from the black hole. 
Also, assuming one can apply the same $M$ multiple
times, one can perform single-query Grover search using $\sim
1/\delta$ applications of $M$ using the methods of \cite{Aaronson_postBQP, Abrams_Lloyd}. More detailed proofs of these results are provided in Appendix A.

We next examine the way in which these results are connected.  First,
assuming one can speed up Grover search, by a generalization of the hybrid argument of \cite{BBBV}, there is a lower bound on the
deviation from unitarity required to achieve the speedup. 
By our
previous results this implies a lower bound on the superluminal
signaling capacity of the map $M$.
More specifically, suppose that one can search an unstructured list of $N$ items using $q$ queries, with possibly non-unitary operations applied between queries. 
Then, the same non-unitary dynamics must be capable of transmitting superluminal signals with channel capacity $C$ using shared entangled states, where
\begin{equation}C = \Omega \left( \left( \frac{\eta}{2 q^2} - \frac{2}{N} \right)^2
\right) \end{equation}
Here $\eta$ is a constant which is roughly $\sim 0.42$. 
In particular, solving
\textsf{NP}-hard problems in polynomial time by unstructured search
would imply superluminal signaling with inverse polynomial channel
capacity.  This can be regarded as evidence against the possibility of
using black hole dynamics to efficiently solve \textsf{NP}-hard
problems of reasonable size.  A proof of this fact is provided in Appendix A.

In the other direction, assuming one can
send a superluminal signal with channel capacity $C$, there is a
lower bound on the deviation from unitarity which was applied. The proof is provided in Appendix A. 
Again by our previous result, this implies one could solve the Grover search problem on a database of size $N$ using a single query and 
\begin{equation}O \left( \frac{\log(N)}{\log(1+C^2)} \right)\end{equation}
applications of the nonlinear map.
Combining these results, this implies that if one can send a superluminal signal
with $n$ applications of $M$, then one can beat Grover's algorithm
with $O(n)$ applications of $M$ as well, and vice versa.  
This shows
that in these models, the resources required to observe an exponential
speedup over Grover search is polynomially related to the resources
needed to send a superluminal signal. Hence an operational version of
the no-signaling principle (such as ``one cannot observe superluminal
signaling in reasonable-sized experiments'') is equivalent to an
operational version of the Grover lower bound (``one cannot observe
violations of the Grover lower bound in reasonable-sized
experiments'').

\subsection{Modification of the Born Rule}

The next continuous modification of quantum mechanics we consider is
modification of the Born rule. Suppose that quantum states evolve
by unitary transformations, but upon measurement one sees outcome $x$
with probability proportional to some function
$f(\alpha_x)$ of the amplitude $\alpha_x$ on $x$. That is, one sees $x$ with probability
\begin{equation}\frac{f(\alpha_x)}{\sum_y f(\alpha_y)}\end{equation}
Note we have added a normalization factor to ensure this induces a valid probability distribution on outcomes.
This is loosely
inspired by Marolf and Polchinski's work \cite{MP15} which suggests
that the ``state-dependence'' resolution of the firewalls paradox
\cite{PR14} gives rise to violations of the Born rule. First, assuming some reasonable conditions on $f$ (namely, that $f$ is differentiable, $f'$ changes signs a finite number of times in $[0,1]$, and the measurement statistics of $f$ do not depend on the normalization of the state),  we must have
$f(\alpha_x)=|\alpha_x|^p$ for some $p$. 
The proof is provided in Appendix B.

Next we study the impact of
such modified Born rules with $p=2+\delta$ for small
$\delta$. Aaronson \cite{Aaronson_postBQP} previously showed that such
models allow for single-query Grover search in polynomial time while incurring a
multiplicative overhead $1/|\delta|$, and also allow for superluminal signaling using shared entangled
states of $\sim1/|\delta|$ qubits.  (His result further generalizes to the harder problem of \emph{counting} the number of solutions to an \textsf{NP}-hard problem, which is a \#\textsf{P}-hard problem). We find that these relationships hold in the opposite directions as well. Specifically, we show if one can send a superluminal signal with an entangled state on $m$
qubits with probability $\epsilon$, then we must have $\delta =\Omega(\epsilon/m)$. 
By the results of Aaronson \cite{Aaronson_postBQP} this implies one can search a list of $N$ items using $O(\frac{m}{\epsilon} \log N)$ time. 
Hence having the ability to send a superluminal signal using $m$ qubits implies the ability to perform an exponential speedup of Grover's algorithm with multiplicative overhead $m$.

In the other direction, if one
can achieve even a constant-factor speedup over Grover's algorithm using a
system of $m$ qubits, we show $|\delta|$ is at least $1/m$ as
well. 
More precisely, by a generalization of the hybrid argument of \cite{BBBV}, if there is an algorithm to search
  an unordered list of $N$ items with $Q$ queries using $m$ qubits, then 
\begin{equation}
\frac{1}{6} \leq \frac{2 Q}{\sqrt{N}} + |\delta| \log (M) +O(\delta^2).
\end{equation}
So if $Q<\sqrt{N}/24$, then we must have $|\delta| \geq \frac{1}{12m}$. 
The proofs of these facts are provided in Appendix B.

Combining these results shows that the number of qubits required
to observe superluminal signaling or even a modest speedup over
Grover's algorithm are polynomially related. Hence one can derive an
operational version of the no-signaling principle from the Grover
lower bound and vice versa. This quantitative result is in some sense
stronger than the result we achieve for the final-state projection model,
because here we require only a mild speedup over Grover search to
derive superluminal signaling.

\subsection{Cloning, Postselection, and Generic Nonlinearities}

We next consider two ``discrete'' modifications of quantum mechanics in
which standard quantum mechanics is supplemented by one additional
operation. We show that both modifications admit both superluminal signaling with O(1) qubits and exponential speedups over Grover search.

First, we consider a model in which one can clone single
qubits. 
This model can be easily seen to admit superluminal signaling using entangled states, as pointed out by Aaronson, Bouland, Fitzsimons and Lee \cite{ABFL}. Indeed, suppose two parties Alice and Bob share the state $\frac{1}{\sqrt{2}}(\ket{00}+\ket{11})$. If Alice measures her half of the state, and Bob clones his state $k$ times and measures each copy in the computational basis, then Bob will either see either $0^k$ or $1^k$ as his output. On the other hand, if Alice does not measure her half of the state, and Bob does the same experiment, his outcomes will be a random string in $\{0,1\}^k$. Bob can distinguish these two cases with an error probability which scales inverse exponentially with $k$, and thus receive a signal faster than light. 
In addition to admitting superluminal signaling with entangled states, this model also allows the solution of
\textsf{NP}-hard problems (and even \#\textsf{P}-hard problems) using a
single query to the oracle. 
This follows by considering the following gadget: given a state $\rho$ on a single qubit, suppose one makes two copies of $\rho$, performs a Controlled-NOT gate between the copies, and discards one of the copies. This is summarized in a circuit diagram in Fig. \ref{fig:gadget}.
\begin{figure}[h]
\includegraphics[width=2.2in]{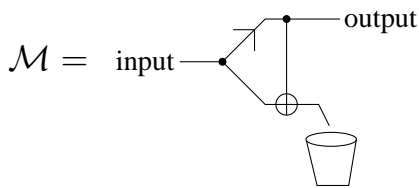}.
\caption{Gadget used to show that cloning allows the poly-time solution of \textsf{NP}-hard problems.} \label{fig:gadget}
\end{figure}

 This performs a non-linear operation $\mathcal{M}$ on the space of density matrices, and following the techniques of Abrams and Lloyd \cite{Abrams_Lloyd}, one can use this operation to ``pry apart" quantum states which are exponentially close using polynomially many applications of the gadget. The proof is provided in Appendix C. This answers an open problem of
\cite{ABFL} about the power of
quantum computers that can clone. Therefore, adding cloning to quantum mechanics allows for both the poly-time solution of \textsf{NP}-hard problems by brute force search, and the ability to efficiently send superluminal signals.

Second, inspired by the final state projection model
\cite{Horowitz_Maldecena}, we consider a model in which one can
postselect on a generic state $\ket{\psi}$ of $n$ qubits. Although
Aaronson \cite{Aaronson_postBQP} previously showed that allowing for
postselection on a single qubit suffices to solve \textsf{NP}-hard and \#\textsf{P}-hard
problems using a single oracle query, this does not immediately imply
that postselecting on a larger state has the same property, because
performing the unitary which rotates $\ket{0}^n$ to $\ket{\psi}$ will
in general require exponentially many gates. Despite this limitation, this model indeed allows the polynomial-time solution of \textsf{NP}-hard problems (as well as \textsf{\#P}-hard problems) and superluminal signaling. 
To see this, first note that given a gadget to postselect on $\ket{\psi}$, one can obtain multiple copies of $\ket{\psi}$ by inputting the maximally entangled state $\sum_i \ket{i}\ket{i}$ into the circuit and postselecting one register on the state $\ket{\psi}$. So consider creating two copies of $\ket{\psi}$, and applying the gadget shown in Figure \ref{postselection_gadget2}, where the bottom
register is postselected onto $\ket{\psi}$, an operation we denote by
$\Qcircuit @C=1em @R=1em { & \gate{\ket{\psi}}  & \qw }$.
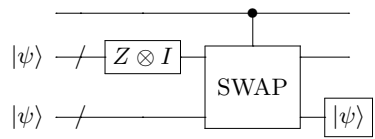
\begin{figure}[b]
\[
\Qcircuit @C=1em @R=1em {
                     & \qw   &\qw   & \ctrl{1}                     & \qw \\
\lstick{\ket{\psi}} &  {/} \qw & \gate{Z\otimes I} &\multigate{1}{\mathrm{SWAP}} & \qw \\
\lstick{\ket{\psi}}  & {/} \qw & \qw & \ghost{\mathrm{SWAP}}        & \gate{\ket{\psi}}
}
\]
\caption{Gadget showing postselection onto generic $\ket{\psi}$ is equivalent to postselection onto $\ket{0}$.}
\label{postselection_gadget2}
\end{figure}
For Haar-random $\ket{\psi}$, one can show the quantity $\bra{\psi}Z\otimes I\ket{\psi}$ is exponentially small, so this gadget simulates postselection on $\ket{0}$ on the first qubit. 
 The complete proof is provided in Appendix D.
Therefore, allowing postselection onto generic states is at least as powerful as allowing postselection onto the state $\ket{0}$, so by Aaronson's results \cite{Aaronson_postBQP} this model admits both superluminal signaling and exponential speedups over Grover search.

In addition, we address an open question from \cite{Abrams_Lloyd}
regarding the computational implications of general nonlinear maps on
pure states. In \cite{Abrams_Lloyd}, Abrams and Lloyd argued that
generic nonlinear maps allow for the solution of \textsf{NP}-hard problems and \textsf{\#P}-hard
problems in polynomial time, except possibly for pathological
examples. In Appendix E, we  prove this result rigorously in the case the map is
differentiable. Thus any pathological examples, if they exist, must
fail to be differentiable. (Here we assume the nonlinearity maps pure states to pure states; as a result it does not subsume our results on quantum computers which can clone, as the cloning operation may map pure states to mixed states. A detailed discussion is provided in Appendix C.) Unfortunately, the action of
general nonlinear maps on subsystems of entangled states are not
well-defined, essentially because they interact poorly with the
linearity of the tensor product. We discuss this in detail in Appendix F. Hence we are unable to connect this
result to signaling in the general case.

\section{Discussion}
\label{discussion}

The central question in complexity theory is which computational problems can be solved efficiently and which cannot. Through experience, computer scientists have found that the most fruitful way to formalize the notion of efficiency is by demanding that the resources, such as time and memory, used to solve a problem must scale at most polynomially with the size of the problem instance (i.e. the size of the input in bits). A widely held conjecture, called the quantum Church-Turing thesis, states that the set of computational problems solvable in-principle with polynomial resources in our universe is equal to BQP, defined mathematically as the set of decision problems answerable using quantum circuits of polynomially many gates \cite{Kaye}. So far, this conjecture has held up remarkably well. Physical processes which conceivably might be more computationally powerful than quantum Turing machines, such as various quantum many-body dynamics of fermions, bosons, and anyons, as well as scattering processes in relativistic quantum field theories, can all be simulated with polynomial overhead by quantum circuits \cite{Aaronson_BosonSampling,FreedmanKitaevWang,Jordan1130,JordanLeePreskill2,JordanLeePreskill3}.

The strongest challenge to the quantum Church-Turing thesis comes from quantum gravity. Indeed, many of the recent quantum gravity models proposed in relation to the black hole firewalls paradox involve nonlinear behavior of wavefunctions \cite{Horowitz_Maldecena,PR14} and thus appear to suggest computational power beyond that of polynomial-size quantum circuits. In particular, the prior work of Abrams and Lloyd suggest that such nonlinearities generically enable polynomial-time solution to NP-hard problems, a dramatic possibility, that standard quantum circuits are not generally expected to admit \cite{Abrams_Lloyd, AaronsonIsland}. Here, we have investigated several models and found a remarkably consistent pattern; in each case, if the modification to quantum mechanics is in a parameter regime allowing polynomial-time solution to NP-hard problems through brute-force search, then it also allows the transmission of superluminal signals through entangled states. Such signaling allows causality to be broken at locations arbitrarily far removed from the vicinity of the black hole, thereby raising serious questions as to the consistency of the models. Thus, the quantum Church-Turing thesis appears to be remarkably robust, depending not in a sensitive way on the complete Hilbert-space formalism of quantum mechanics, but rather derivable from more foundational operational principles such as the impossibility of superluminal signaling. Some more concrete conjectures on these lines are discussed in Appendix G.
\section*{Acknowledgments}We thank Patrick Hayden, Daniel Harlow,
David Meyer, Andrew Childs and Debbie Leung for useful discussions.  Portions of
this paper are a contribution of NIST, an agency of the US government,
and are not subject to US copyright.  This material is based upon work supported by the DuBridge Postdoctoral Fellowship, by the Institute for Quantum Information and Matter, an NSF Physics Frontiers Center (NFS Grant PHY-1125565) with support of the Gordon and Betty Moore Foundation (GBMF-12500028), by the U.S. Department of Energy, Office of Science, Office of High Energy Physics, under Award Number DE-SC0011632, by the NSF Graduate Research Fellowship under grant
no. 1122374, and by the NSF Alan T. Waterman award under grant no. 1249349.  
N. B. and A. B. would also
like to thank QuICS for their hospitality during the completion of
this project.

\onecolumngrid
\appendix

\section*{Appendix A: Final-State Projection}
\stepcounter{section}
\label{fsp}

Recent developments, particularly the AMPS firewall argument
\cite{AMPS}, have generated renewed interest in models of black hole
physics in which quantum mechanics is modified. Here, we explore some
difficulties associated one such scheme, namely the
Horowitz-Maldecena final state projection model
\cite{Horowitz_Maldecena}. In this model, black hole singularities are
thought of as boundaries to spacetime with associated boundary
conditions on the quantum wavefunction \cite{Horowitz_Maldecena}. That
is, at the singularity, the wavefunction becomes projected onto a
specific quantum state. (This can be thought of as a projective
measurement with postselection.)

If one prepares infalling matter in a chosen initial quantum state
$\ket{\psi} \in V$, allows it to collapse into a black hole, and then
collects all of the the Hawking radiation during the black hole
evaporation, one is left with a new quantum state related to the
original by some map $S:V \to V$. (We assume that black holes do not
alter the dimension of the Hilbert space. Standard quantum mechanics
and the Horowitz-Maldecena proposal share this feature.) Within
standard quantum mechanics, all such $S$ correspond to multiplication
by a unitary matrix, and hence the term $S$-matrix is used. If one
instead drops matter into an existing black hole and collects part of
the outgoing Hawking radiation, one is considering an open quantum
system. We leave the analysis of this more general scenario to future
work.

It is possible for the Horowitz-Maldecena final state projection model
to induce a perfectly unitary process $S$ for the black hole. However,
as pointed out in \cite{Gottesman_Preskill}, interactions between the
collapsing body and infalling Hawking radiation inside the event
horizon generically induce deviations from unitarity. In this case,
the action $S$ of the black hole is obtained by applying some linear
but not unitary map $M$, and then readjusting the norm of the quantum
state back to one\footnote{Some interpret the final state projection
model as inducing a unitary map for observers who stay outside the
event horizon, while inducing a non-unitary map for infalling
observers \cite{danharlow}. Under this interpretation, our arguments
still apply to an infalling observer in the context of a large black
hole.}. Correspondingly, if a subsystem of an entangled state is
collapsed into a black hole and the Hawking radiation is collected
then the corresponding transformation is $M \otimes \id$ followed by
an adjustment of the normalization back to 1. Thus, aside from its
interest as a potential model for black holes, the Horowitz-Maldecena
model provides an interesting example of nonlinear quantum mechanics
in which subsystem structure remains well-defined (i.e. the issues
described in Appendix F do not arise).

In sections \ref{AB} and \ref{BA} we show that if Alice has access to
such a black hole and has foresightfully shared entangled states with
Bob, then Alice can send instantaneous noisy signals to Bob and
vice-versa independent of their spatial separation. We quantify the
classical information-carrying capacity of the communication channels
between Alice and Bob and find that they vanish only quadratically
with the deviation from unitarity of the black hole dynamics, as
measured by the deviation of the condition number of $M$ from
one. Hence, unless the deviation from unitarity is negligibly small,
detectable causality violations can infect the entirety of
spacetime. Furthermore, the bidirectional nature of the communication
makes it possible in principle for Alice to send signals into her own
past lightcone, thereby generating grandfather paradoxes.

In section \ref{grover} we consider the use of the black hole
dynamical map $S$ to speed up Grover's search algorithm
\cite{Grover}. We find a lower bound on the condition number of $M$ as
a function of the beyond-Grover speedup. By our results of sections
\ref{AB} and \ref{BA} this in turn implies a lower bound on the
superluminal signaling capacity induced by the black hole. In section
\ref{HMsignalingsearch} we prove the other direction: assuming one can
signal superluminally we derive a lower bound on the condition number
of $M$, which in turn implies a super-Grover speedup\footnote{By a
``super-Grover speedup'', we mean an algorithm which searches an
unstructured $N$-element list using fewer queries than Grover's
algorithm.}. We find that the
black-box solution of \textsf{NP}-hard problems in polynomial time
implies superluminal signaling with inverse polynomial capacity and
vice versa.

\subsection{Communication from Alice to Bob}
\label{AB}

\begin{theorem} Suppose Alice has access to a black hole described by
  the Horowitz-Maldecena final state projection model. Let $M$ be the
  linear but not necessarily unitary map describing the dynamics of
  the black hole. The non-unitarity of $M$ is quantified by $\delta =
  1-\kappa$, the deviation of its condition number from one. Alice can
  transmit instantaneous signals to Bob by choosing to drop her half
  of a shared entangled state into the black hole or not. The capacity
  of the resulting classical communication channel from Alice to Bob
  is at least
\[
  C \geq \frac{3}{8 \ln 2} \delta^2.
\]
\end{theorem} 

\begin{proof}
We prove the lower bound on the channel capacity $C$ by exhibiting an
explicit protocol realizing it. Suppose the black hole acts on a
$d$-dimensional Hilbert space and correspondingly $M$ is a $d \times
d$ matrix. Then, $M$ has a singular-value decomposition given by
\begin{equation}
M = \sum_{i=0}^{d-1} \ket{\psi_i} \lambda_i \bra{\phi_i}
\end{equation}
with
\begin{equation}
\label{ortho}
\braket{\psi_i}{\psi_j} = \braket{\phi_i}{\phi_j} = \delta_{ij}.
\end{equation}
and $\lambda_0,\ldots,\lambda_{d-1}$ all real and nonnegative. We can
choose our indexing so that $\lambda_0$ is the smallest singular value 
and $\lambda_1$ is largest singular value. Now, suppose Alice and Bob
share the state
\begin{equation}
\frac{1}{\sqrt{2}} \left( \ket{\phi_0} \ket{0} + \ket{\phi_1} \ket{1} \right).
\end{equation}
Here $\ket{0}$ and $\ket{1}$ refer to Bob's half of the entangled
state, which can be taken to be a qubit. If Alice wishes to transmit
the message ``0'' to Bob she does nothing. If she wishes to transmit
the message ``1'' to Bob she applies the black hole dynamical map $S$
to her half of the state. In other words, Alice drops her half of the state into the black hole, and waits for the black hole to evaporate. Correspondingly, one applies $M \otimes \id$
to the above state, yielding the unnormalized state
\begin{equation}
\frac{\lambda_0}{\sqrt{2}} \ket{\psi_0} \ket{0} +
\frac{\lambda_1}{\sqrt{2}} \ket{\psi_1} \ket{1}.
\end{equation}
After normalization, this becomes:
\begin{equation}
\frac{\lambda_0}{\sqrt{\lambda_0^2 + \lambda_1^2}} \ket{\psi_0} \ket{0}
+ \frac{\lambda_1}{\sqrt{\lambda_0^2 + \lambda_1^2}} \ket{\psi_1} \ket{1}.
\end{equation}
Thus, recalling \eq{ortho}, Bob's reduced density matrix in this case
is
\begin{equation}
\rho_1 = \frac{\lambda_0^2}{\lambda_0^2 + \lambda_1^2} \ket{0} \bra{0} +
\frac{\lambda_1^2}{\lambda_0^2+\lambda_1^2} \ket{1} \bra{1},
\end{equation}
whereas in the case that Alice's message was ``0'' his reduced density
matrix is
\begin{equation}
\rho_0 = \frac{1}{2} \ket{0} \bra{0} + \frac{1}{2} \ket{1} \bra{1}.
\end{equation}
If $M$ is non-unitary then $\lambda_1 \neq \lambda_0$ and thus the
trace distance between these density matrices is
nonzero. Consequently, $\rho_1$ is distinguishable from $\rho_0$
and some fraction of a bit of classical information has been
transmitted to Bob. 

More quantitatively, one sees that Bob's optimal measurement is in the
computational basis, in which case Alice and Bob are communicating
over a classical binary asymmetric channel. Specifically, if Alice
transmits a 0, the probability of bit-flip error is $\epsilon_0 = 1/2$
whereas if Alice transmits a 1, the probability of bit-flip error is 
\begin{equation}
\label{ep1lam}
\epsilon_1 = \frac{\lambda_0^2}{\lambda_0^2 + \lambda_1^2}.
\end{equation}
A standard calculation
(see \emph{e.g} \cite{Leslie12}) shows that the classical capacity of
this channel is
\begin{equation}
C = h \left( \frac{1}{1+z} \right) - \frac{\log_2(z)}{1+z} + \epsilon_0
\log_2(z) - h(\epsilon_0),
\end{equation}
where
\begin{equation}
z = 2^{\frac{h(\epsilon_1)-h(\epsilon_0)}{1-\epsilon_1-\epsilon_0}}
\end{equation}
and $h$ is the binary entropy
\begin{equation}
h(p) = - p \log_2(p) -(1-p) \log_2(1-p). 
\end{equation}
Specializing to $\epsilon_0 = \frac{1}{2}$ simplifies the expression to
\begin{equation}
C = h \left( \frac{1}{1+y} \right) - \frac{\log_2(y)}{1+y} +
\frac{1}{2} \log_2 (y) - 1
\end{equation}
where
\begin{equation}
y = 2^{\frac{h(\epsilon_1) - 1}{1/2-\epsilon_1}}.
\end{equation}
Lastly, we consider the limiting case $\epsilon_1 = \frac{1}{2} -
\Delta$ for $\Delta \ll 1$. In this limit, we get by Taylor
expansion that 
\begin{equation}
\label{epscap}
C = \frac{3}{2 \ln 2} \Delta^2 + O(\delta^3).
\end{equation}
By \eq{ep1lam}, $\Delta = \frac{1}{2} (1-\kappa) + O((1-\kappa)^2)$,
which completes the proof.
\end{proof}

\subsection{Communication from Bob to Alice}
\label{BA}
\begin{theorem} Suppose Alice has access to a black hole described by
  the Horowitz-Maldecena final state projection model. Let $M$ be the
  linear but not necessarily unitary map describing the dynamics of
  the black hole. The non-unitarity of $M$ is quantified by $\delta =
  1-\kappa$, the deviation of its condition number from one. Bob can
  transmit instantaneous signals to Alice by choosing to measure his
  half of a shared entangled state or not. The capacity
  of the resulting classical communication channel from Bob to
  Alice is at least 
\[
  C \geq \frac{3}{8 \ln 2} \delta^2.
\]
\end{theorem}

\begin{proof}
Suppose again that Alice and Bob share the state $\frac{1}{\sqrt{2}}
\left( \ket{\phi_0} \ket{0} + \ket{\phi_1} \ket{1} \right)$. If Bob
wishes to transmit the message ``0'' he does nothing, whereas if he
wishes to transmit the message ``1'' he measures his half of the
entangled state in the $\{\ket{0},\ket{1}\}$ basis. Then, Alice
applies the black hole dynamical map $S$ to her half of the state\footnote{That is, Alice drops her half of the shared state into the black hole, and waits for the black hole to evaporate.}, and
then performs a projective measurement in the basis
$\{\ket{\psi_1},\ldots,\ket{\psi_d}\}$. We now show that this
procedure transmits a nonzero amount of classical information from Bob
to Alice unless $\lambda_0 = \lambda_1$, in which case $M$ is unitary.

In the case that Bob does nothing, the post-black hole state is again
\begin{equation}
\frac{\lambda_0}{\sqrt{\lambda_0^2 + \lambda_1^2}} \ket{\psi_0} \ket{0}
+ \frac{\lambda_1}{\sqrt{\lambda_0^2 + \lambda_1^2}} \ket{\psi_1} \ket{1}.
\end{equation}
Thus, Alice's post-black-hole reduced density matrix is
\begin{equation}
\frac{\lambda_0^2}{\lambda_0^2 + \lambda_1^2} \ket{\psi_0}
\bra{\psi_0} + \frac{\lambda_1^2}{\lambda_0^2 + \lambda_1^2} \ket{\psi_1}
\bra{\psi_1}.
\end{equation}
Alice's measurement will consequently yield the following probability
distribution, given that Bob's message was ``0'':
\begin{eqnarray}
p(0|0) & = & \frac{\lambda_0^2}{\lambda_0^2 + \lambda_1^2} \\
p(1|0) & = & \frac{\lambda_1^2}{\lambda_0^2 + \lambda_1^2}.
\end{eqnarray}
Now, suppose Bob's message is ``1''. Then, his measurement outcome
will be either $\ket{0}$ or $\ket{1}$ with equal probability. We must
analyze these cases separately, since the connection between ensembles
of quantum states and density matrices is not preserved under
nonlinear transformations\footnote{Elsewhere we have used density
  matrices, but only after the application of the nonlinear
  operation.}. If he gets outcome zero, then Alice holds the pure
state $\ket{\phi_0}$, which gets transformed to $\ket{\psi_0}$ by the
action of the black hole. If Bob gets outcome one, then Alice holds
$\ket{\phi_1}$, which gets transformed to $\ket{\psi_1}$ by the action
of the black hole. Hence, Alice's measurement samples from the
following distribution given that Bob's message was ``1'':
\begin{eqnarray}
p(0|1) & = & 1/2 \\
p(1|1) & = & 1/2.
\end{eqnarray}
Hence, the information transmission capacity from Bob to Alice using
this protocol is the same as the Alice-to-Bob capacity calculated in
section \ref{AB}.
\end{proof}

\subsection{Super-Grover Speedup implies Superluminal Signaling}
\label{grover}

\begin{theorem} Suppose one has access to one or more black holes
  described by the Horowitz-Maldecena final state projection model. If
  the non-unitary dynamics induced by the black hole(s) allow the
  solution of a Grover search problem on a database of size $N$ using $q$
  queries then the same non-unitary dynamics could be used transmit
  instantaneous signals by applying them to half of an entangled
  state. The capacity of the resulting classical communication channel
  (bits communicated per use of the nonlinear dynamics) is at least
\[
C = \Omega \left( \left( \frac{\eta}{2 q^2} - \frac{2}{N} \right)^2
\right)
\]
in the regime $0 < \frac{\eta}{2 q^2} - \frac{2}{N} \ll 1$, where
$\eta = (\sqrt{2} - \sqrt{2 - \sqrt{2}})^2 \simeq 0.42$.
\end{theorem}
 
\begin{proof}
Let $V$ be the set of normalized vectors in the Hilbert space
$\mathbb{C}^{N}$. We will let $S:V \to V$ denote the nonlinear map
that a black hole produces by applying the matrix $M$ and then
readjusting the norm of the state to one. We will not assume that all
black holes are identical, and therefore, each time we interact with a
black hole we may have a different map. We denote the transformation
induced on the $k\th$ interaction by $S_k:V \to V$. We treat $S_k$ as
acting on the same state space for all $k$, but this is not actually a
restriction because we can simply take this to be the span of all the
Hilbert spaces upon which the different maps act.

Now suppose we wish to use the operations $S_1,S_2,\ldots$ to speed up
Grover search. Let $x \in \{0,\ldots,N-1\}$ denote the solution to the
search problem on $\{0,\ldots,N-1\}$. The corresponding unitary oracle
on $\mathbb{C}^N$ is\footnote{An alternate definition is to use a
  bit-flip oracle $U_x \ket{y}\ket{z} = \ket{y} \ket{z \oplus f(y)}$
    where $f(y) = 1$ if $y=x$ and $f(y)=0$ otherwise. This choice is
    irrelevant since the phase flip oracle can be simulated using a
    bit-flip oracle if the output register is initialized to
    $(\ket{0}-\ket{1})/\sqrt{2}$, and a bit flip oracle can be
    simulated using a controlled-phase-flip oracle.}
\begin{equation}
O_x = \id - 2 \ket{x}\bra{x}.
\end{equation}
The most general algorithm to search for $x$ is of the form
\begin{equation}
\label{finalstate}
S_q  O_x  \ldots S_2 O_x S_1 O_x \ket{\psi_0}
\end{equation}
followed by a measurement. Here $\ket{\psi_0}$ is any $x$-independent
quantum state on $\mathbb{C}^N$, and $S_k$ is any transformation that
can be achieved on $\mathbb{C}^N$ by any sequence of unitary
operations and interactions with black holes. Note that our
formulation is quite general and includes the case that multiple
non-unitary interactions are used after a given oracle
query, as is done in \cite{Abrams_Lloyd}. Also, for some $k$, $S_k$
may be purely unitary. For example,
one may have access to only a single black hole, and the rest of the
iterations of Grover's algorithm must be done in the ordinary unitary
manner. If the final measurement on the state described in
\eq{finalstate} succeeds in identifying $x$ with high probability for
all $x \in \{0,\ldots,N-1\}$ then we say the query complexity of
Grover search using the black hole is at most $q$.

We now adapt the proof of the $\Omega(\sqrt{N})$ quantum query lower
bound for Grover search that was given in\footnote{Our notation is
  based on the exposition of this proof given in
  \cite{Nielsen_Chuang}.} \cite{BBBV} to show that any improvement in
the query complexity for Grover search implies a corresponding lower
bound on the ability of $S_k$ for some $k \in \{1,\ldots,q\}$ to ``pry
apart'' quantum states. This then implies a corresponding lower bound
on the rate of a superluminal classical information transmission
channel implemented using $S_k$.

The sequence of quantum states obtained in the algorithm
\eq{finalstate} is
\begin{eqnarray}
\ket{\psi_0^x} & = & \ket{\psi_0} \nonumber \\
\ket{\phi_1^x} & = & O_x \ket{\psi_0} \nonumber \\
\ket{\psi_1^x} & = & S_1 O_x \ket{\psi_0} \label{stateseq} \\
\ket{\phi_2^x} & = & O_x S_1 O_x \ket{\psi_0} \nonumber \\
\ket{\psi_2^x} & = & S_2 O_x S_1 O_x \ket{\psi_0} \nonumber \\
               & \vdots & \nonumber \\
\ket{\psi_q^x} & = & S_q O_x \ldots S_1 O_x \ket{\psi_0}. \nonumber
\end{eqnarray}
Let
\begin{eqnarray}
\ket{\psi_k} & = & S_k S_{k-1} \ldots S_1 \ket{\psi_0} \\
C_k & = & \sum_{x=0}^{N-1} \| \ket{\phi_k^x} - \ket{\psi_{k-1}} \|^2
\label{ck} \\
D_k & = & \sum_{x=0}^{N-1} \| \ket{\psi_k^x} - \ket{\psi_k} \|^2. \label{dk}
\end{eqnarray}
$\ket{\psi_k}$ can be interpreted as the state which would have been
obtained after the $k\th$ step of the algorithm with no oracle queries
(or of the Grover search problem lacked a solution).

Now, assume that for all $x \in \{0,\ldots,N-1\}$ the search algorithm
succeeds after $q$ queries in finding $x$ with probability at least
$\frac{1}{2}$. Then,
\begin{equation}
| \langle x | \psi_q^x \rangle |^2 \geq \frac{1}{2} \quad \forall x
\in \{0,\ldots,N-1\}
\end{equation}
which implies
\begin{equation}
D_q \geq \eta N, \label{bigomega}
\end{equation}
with $\eta = (\sqrt{2}-\sqrt{2-\sqrt{2}})^2 \simeq 0.42$, as shown in
\cite{BBBV} and discussed in \cite{Nielsen_Chuang}. By \eq{stateseq},
\eq{ck}, and \eq{dk},
\begin{eqnarray}
C_k & = & \sum_{x=0}^{N-1} \| O_x \ket{\psi_{k-1}^x} -
\ket{\psi_{k-1}} \|^2 \\
 & \leq & D_{k-1} + 4 \sqrt{D_{k-1}} + 4, \label{ckbound1}
\end{eqnarray}
where the above inequality is obtained straightforwardly using the
triangle and Cauchy-Schwarz inequalities.

Next, let
\begin{equation}
R_k = D_k - C_k. \label{rk}
\end{equation}
Thus, by \eq{stateseq}, \eq{ck}, and \eq{dk}, 
\begin{equation}
R_k = \sum_{x=0}^{N-1} \| S_k \ket{\phi_k^x} - S_k \ket{\psi_{k-1}} \|^2
 - \sum_{x=0}^{N-1} \| \ket{\phi_k^x} - \ket{\psi_{k-1}} \|^2. \label{rkunpack}
\end{equation}
Hence, one sees that $R_k$ is some measure of the ability of $S_k$ to
``pry apart'' quantum states. (In ordinary quantum mechanics $S_k$ would
be unitary and hence $R_k$ would equal zero.)

Combining \eq{rk} and \eq{ckbound1} yields
\begin{equation}
D_k \leq R_k + D_{k-1} + 4 \sqrt{D_{k-1}} + 4. \label{recurrence1}
\end{equation}
Let
\begin{equation}
B = \max_{1 \leq k \leq q} R_k. \label{Bdef}
\end{equation}
Then \eq{recurrence1} yields the simpler inequality
\begin{equation}
D_k \leq B + D_{k-1} + 4 \sqrt{D_{k-1}} + 4. \label{recurrence2}
\end{equation}
By \eq{stateseq} and \eq{dk},
\begin{equation}
D_0 = 0. \label{init}
\end{equation}
By an inductive argument, one finds that \eq{recurrence2} and \eq{init} imply
\begin{equation}
D_k \leq (4+B) k^2. \label{payoff}
\end{equation}
Combining \eq{payoff} and \eq{bigomega} yields
\begin{equation}
(4+B)q^2 \geq \eta N,
\end{equation}
or in other words
\begin{equation}
B \geq \frac{\eta N}{q^2} -4.
\end{equation}
Thus, by \eq{Bdef} and \eq{rkunpack}, there exists some $k \in
\{1,\ldots,q\}$ such that
\begin{equation}
\sum_{x=0}^{N-1} \left( \| S_k \ket{\phi_k^x} - S_k \ket{\psi_{k-1}}
  \|^2 - \| \ket{\phi_k^x} - \ket{\psi_{k-1}} \|^2 \right) \geq
\frac{\eta N}{q^2} -4
\end{equation}
Hence, there exists some $x \in \{0,\ldots,N-1\}$ such that
\begin{equation}
\| S_k \ket{\phi_k^x} - S_k \ket{\psi_{k-1}} \|^2 - 
\| \ket{\phi_k^x} - \ket{\psi_{k-1}} \|^2 \geq 
\frac{\eta}{q^2} - \frac{4}{N}. \label{normpry}
\end{equation}
To simplify notation, define
\begin{eqnarray}
\ket{A} & = & \ket{\phi_k^x} \\
\ket{B} & = & \ket{\psi_{k-1}} \\
\ket{A'} & = & S_k \ket{\phi_k^x} \\
\ket{B'} & = & S_k \ket{\psi_{k-1}}.
\end{eqnarray}
Then \eq{normpry} becomes
\begin{equation}
\| \ket{A'} - \ket{B'} \|^2 - \| \ket{A} - \ket{B} \|^2 \geq
\frac{\eta}{q^2} - \frac{4}{N}. \label{normpryAB}
\end{equation}
Recalling that $\| \ket{\psi} \| = \sqrt{\langle \psi | \psi
  \rangle}$, \eq{normpryAB} is equivalent to
\begin{equation}
\label{dotpry}
\re \langle A | B \rangle - \re \langle A' | B' \rangle \geq \epsilon
\end{equation}
with
\begin{equation}
\label{epsdef}
\epsilon = \frac{\eta}{2 q^2} - \frac{2}{N}. 
\end{equation}

Next we will show that, within the framework of final-state projection
models, \eq{dotpry} implies that Alice can send a polynomial fraction
of a bit to Bob or vice versa using preshared entanglement and a
single application of black hole dynamics. Recall that, within the
final state projection model,
\begin{eqnarray}
\ket{A'} & = & \frac{M \ket{A}}{\sqrt{\bra{A} M^\dag M \ket{A}}} \\
\ket{B'} & = & \frac{M \ket{B}}{\sqrt{\bra{B} M^\dag M \ket{B}}}
\end{eqnarray}
Thus, \eq{dotpry} is equivalent to
\begin{equation}
\re \left[ \bra{A} \left( \id - \frac{M^\dag M}
{\sqrt{\bra{A} M^\dag M \ket{A} \bra{B} M^\dag M \ket{B}}} \right)
\ket{B} \right] \geq \epsilon
\end{equation}
Hence,
\begin{equation}
\label{Mnorm}
\left\| \id -  \frac{M^\dag M}
{\sqrt{\bra{A} M^\dag M \ket{A} \bra{B} M^\dag M \ket{B}}} \right\| \geq \epsilon.
\end{equation}
Again using $\lambda_0$ to denote the smallest singular value of $M$ 
and $\lambda_1$ to denote the largest, we see that, assuming
$\epsilon$ is nonnegative, \eq{Mnorm} implies either\\
\textbf{Case 1:}
\begin{equation}
\frac{\lambda_1^2}{\sqrt{\bra{A} M^\dag M \ket{A} \bra{B} M^\dag M
    \ket{B}}} \geq 1 + \epsilon,
\end{equation}
which implies
\begin{equation}
\label{direct}
\frac{\lambda_1^2}{\lambda_0^2} \geq 1 + \epsilon,
\end{equation}
or\\\
\textbf{Case 2:}
\begin{equation}
\frac{\lambda_0^2}{\sqrt{\bra{A} M^\dag M \ket{A} \bra{B} M^\dag M
    \ket{B}}} \leq 1 - \epsilon,
\end{equation}
which implies
\begin{equation}
\label{recip}
\frac{\lambda_0^2}{\lambda_1^2} \leq 1 - \epsilon.
\end{equation}
Examining \eq{dotpry}, one sees that $\epsilon$ can be at most 2. If 
$0 \leq \epsilon \leq 1$ then \eq{recip} implies \eq{direct}. If 
$1 < \epsilon \leq 2$ then case 2 is impossible. Hence, for any
nonnegative $\epsilon$ one obtains \eq{direct}. Hence, by the results
of sections \ref{AB} and \ref{BA}, Alice and Bob can communicate in
either direction through a binary asymmetric channel whose bitflip
probabilities $\epsilon_0$ for transmission of zero and $\epsilon_1$
for transmission of one are given by
\begin{eqnarray}
\epsilon_0 & = & \frac{1}{2} \label{epnaught} \\
\epsilon_1 & = & \frac{\lambda_0^2}{\lambda_0^2+\lambda_1^2} \leq
\frac{1}{2+\epsilon}. \label{epone_raw}
\end{eqnarray}
For $0 \leq \epsilon \leq 2$, $\frac{1}{2+\epsilon} \leq \frac{1}{2} -
\frac{\epsilon}{8}$. Thus, \eq{epone_raw} implies the following more
convenient inequality
\begin{equation}
\label{epone}
\epsilon_1 \leq \frac{1}{2} - \frac{\epsilon}{8}.
\end{equation}
In section \ref{AB} we calculated that the channel capacity in the
case that $\epsilon_0 = \frac{1}{2}$ and $\epsilon_1 = \frac{1}{2} -
\delta$ is $\Omega(\delta^2)$ for $\delta \ll 1$. Thus,
\eq{epnaught} and \eq{epone} imply a channel capacity in either
direction of
\begin{equation}
\label{final_capacity}
C = \Omega \left( \left( \frac{\eta}{2 q^2} - \frac{2}{N} \right)^2
\right)
\end{equation}
in the regime $0 < \frac{\eta}{2 q^2} - \frac{2}{N} \ll 1$.
\end{proof}

The above scaling of the superluminal channel capacity with Grover
speedup shows that polynomial speedup for small instances or
exponential speedup for large instances imply $1/\mathrm{poly}$
superluminal channel capacity. In particular, to solve \textsf{NP} in
polynomial time without exploiting problem structure we would need $q
\propto \log^c N$ for some constant $c$. In this setting $N = 2^n$
where $n$ is the size of the witness for the problem in NP. In this
limit, \eq{final_capacity} implies instantaneous signaling channels in
each direction with capacity at least
\begin{equation}
\label{Cval}
C = \Omega \left( \frac{1}{\log^{4c} N} \right) =  \Omega \left(
\frac{1}{n^{4c}} \right).
\end{equation}
If we assume that superluminal signaling capacity is limited to some
negligibly small capacity $C \leq \epsilon$ then, by \eq{Cval}, NP-hard
problems cannot be solved by unstructured search in time scaling
polynomially with witness size (specifically $n^c$ for some constant
$c$) except possibly for unphysically large instances with 
$n = \Omega \left( \left( \frac{1}{\epsilon} \right)^{\frac{1}{4c}} \right)$.

\subsection{Signaling implies Super-Grover Speedup}
\label{HMsignalingsearch}

In sections \ref{AB} and \ref{BA} we showed that if final-state
projection can be used to speed up Grover search it can also be used
for superluminal signaling. In this section we show the
converse. Unlike in section \ref{grover}, we here make the
assumption that we can make multiple uses of the same non-unitary map
$S$ (just as other quantum gates can be used multiple times without
variation). Since signaling cannot be achieved by performing unitary
operations on entangled quantum degrees  of freedom, superluminal
signaling implies non-unitarity.  Furthermore, as shown in Appendix F, iterated application of 
any nonlinear but differentiable map allows the Grover search problem
to be solved with only a single oracle query. The nonlinear maps that
arise in final-state projection models are differentiable (provided
$M$ is invertible), and thus within the final-state projection
framework signaling implies single-query Grover search. In the
remainder of this section we quantitatively investigate how many
iterations of the nonlinear map are needed to achieve single-query
Grover search, as a function of the superluminal signaling capacity. We
find that unless the signaling capacity is exponentially small,
logarithmically many iterations suffice. Specifically, our main result
of this section is the following theorem.

\begin{theorem}
Suppose Alice has access to a linear but not necessarily unitary maps
on quantum states, as can arise in the Horowitz-Maldecena final
state projection model. Suppose she achieves instantaneous
classical communication capacity of $C$ bits transmitted per use of
nonunitary dynamics. Then she could solve the Grover search problem on
a database of size $N$ using a single query and 
$O \left( \frac{\log(N)}{\log(1+C^2)} \right)$ applications of the
available nonunitary maps.
\end{theorem}

\begin{proof}
Suppose Alice has access to black hole(s) and Bob does not. Alice will
use this to send signals to Bob using some shared entangled state
$\ket{\psi}_{AB}$. Her most general protocol is to apply some
map $M_0$ to her half of the state if she wishes to transmit a zero
and some other map $M_1$ if she wishes to transmit a one. (As a special
case, $M_0$ could be the identity.) Here, per the final state
projection model, $M_0$ and $M_1$ are linear but not necessarily
unitary maps, and normalization of quantum states is to be adjusted
back to one after application of these maps. The possible states
shared by Alice and Bob given Alice's two possible messages are
\begin{eqnarray}
\ket{\psi_0}_{AB} & \propto & M_0 \ket{\psi}_{AB} \\
\ket{\psi_1}_{AB} & \propto & M_1 \ket{\psi}_{AB}.
\end{eqnarray}

The signaling capacity is determined by the distinguishability of the
two corresponding reduced density matrices held by Bob
\begin{eqnarray}
\rho_0 & = & \mathrm{Tr}_A \left[ \ket{\psi_0}_{AB} \right] \\
\rho_1 & = & \mathrm{Tr}_A \left[ \ket{\psi_1}_{AB} \right].
\end{eqnarray}
We can define
\begin{equation}
\ket{\psi'} \propto M_0 \ket{\psi}_{AB}
\end{equation}
in which case
\begin{equation}
\ket{\psi_1}_{AB} \propto M_1 M_0^{-1} \ket{\psi'}.
\end{equation}
(We normalize $\ket{\psi'}$ so that $\braket{\psi'}{\psi'} = 1$.)
Thus, the signaling capacity is determined by the distinguishability
of
\begin{eqnarray}
\rho_0 & = & \mathrm{Tr}_A \left[ \ket{\psi'} \right] \\
\rho_1 & = & \mathrm{Tr}_A \left[ \frac{1}{\eta} M \ket{\psi'} \right]
\end{eqnarray}
where
\begin{eqnarray}
M & = & M_1 M_0^{-1} \\
\eta & = & \sqrt{\bra{\psi'} M^\dag M \ket{\psi'}}.
\end{eqnarray}

We have thus reduced our analysis to the case that Alice applies some
non-unitary map $M$ to her state if she wants to transmit a one and
does nothing if she wants to transmit a zero. We will next obtain a
lower bound $\kappa_{\min}$ on the condition number of $M$ as a function of
the signaling capacity from Alice to Bob. This then implies that one
of $M_0, M_1$ has a condition number at least $\sqrt{\kappa_{\min}}$
for the general case.

Suppose that $M$ has the following singular value decomposition
\begin{equation}
\label{svd}
M = \sum_i \lambda_i \ket{\psi_i} \bra{\phi_i}.
\end{equation}
We can express $\ket{\psi'}$ as
\begin{equation}
\ket{\psi'} = \sum_{i,j} \alpha_{ij} \ket{\phi_i} \ket{B_j}
\end{equation}
where $\ket{\phi_1}, \ket{\phi_2}, \ldots$ is the basis determined by
the singular value decomposition \eq{svd} and $\ket{B_1}, \ket{B_2},
\ldots$ is the basis Bob will perform his measurement in when he tries
to extract Alice's message. If Alice wishes to transmit one then she
applies $M$ yielding
\begin{equation}
\ket{\psi_1} \propto \sum_{i,j} \lambda_i \alpha_{ij} \ket{\psi_i} \ket{B_j}.
\end{equation}
So
\begin{eqnarray}
\rho_0 & = & \sum_{i,j,k} \alpha_{ij} \alpha^*_{ik} \ket{B_j}\bra{B_k}\\
\rho_1 & = & \sum_{i,j,k} \frac{\lambda_i^2}{\eta} \alpha_{ij}
\alpha^*_{ik} \ket{B_j} \bra{B_k}.
\end{eqnarray}
Consequently, Bob's measurement will yield a sample from the following
probability distributions conditioned on Alice's message.
\begin{eqnarray}
p(j|0) & = & \sum_i |\alpha_{ij}|^2 \\
p(j|1) & = & \sum_i \frac{\lambda_i^2}{\eta} |\alpha_{ij}|^2.
\end{eqnarray}
The total variation distance between these distributions, which
determines the capacity of the superluminal channel is
\begin{equation}
\label{deltadef}
\Delta = \frac{1}{2} \sum_j |p(j|0)-p(j|1)| = \frac{1}{2} \sum_j \left| \sum_i
|\alpha_{ij}|^2 \left( 1 - \frac{\lambda_i^2}{\eta} \right) \right|.
\end{equation}

From a given value of this total variation distance we wish to derive
a lower bound on the condition number of $M$, that is, the ratio of
the largest singular value to the smallest. Applying the triangle
inequality to \eq{deltadef} yields
\begin{equation}
\label{ineq1}
\Delta \leq \frac{1}{2} \sum_{ij} |\alpha_{ij}|^2 \left| 1 -
  \frac{\lambda_i^2}{\eta} \right|.
\end{equation}
Because $\alpha_{ij}$ are amplitudes in a normalized quantum state,
\begin{equation}
p(i) = \sum_j |\alpha_{ij}|^2
\end{equation}
is a probability distribution. We can thus rewrite \eq{ineq1}
as
\begin{eqnarray}
\label{ineq2a}
\Delta & \leq & \frac{1}{2} \sum_i p(i) \left| 1 -
  \frac{\lambda_i^2}{\eta} \right| \\
& \leq & \frac{1}{2} \max_i \left| 1 - \frac{\lambda_i^2}{\eta}
\right|. \label{ineq2b}
\end{eqnarray}
In keeping with the notation from previous sections, we let
$\lambda_0$ denote the smallest singular value of $M$ and $\lambda_1$
the largest. Thus, \eq{ineq2b} yields
\begin{equation}
\label{ineq3}
\Delta \leq \frac{1}{2} \max \left\{ 1-\frac{\lambda_0^2}{\eta},
  \frac{\lambda_1^2}{\eta} - 1 \right\}.
\end{equation}
Similarly,
\begin{eqnarray}
\eta & = & \sum_{jk} |\alpha_{jk}|^2 \lambda_j^2 \\
     & = & p(j) \lambda_j^2 \\
     & \in & [\lambda_0^2, \lambda_1^2]. \label{etarange}
\end{eqnarray}
Applying \eq{etarange} to \eq{ineq3} yields
\begin{equation}
\label{deltakappa}
\Delta \leq \frac{1}{2} \max \left\{ 1-\frac{\lambda_0^2}{\lambda_1^2},
  \frac{\lambda_1^2}{\lambda_0^2} - 1 \right \}.
\end{equation}

As shown in section \ref{apptvd}, the channel capacity $C$ is related
to the total variation distance $\Delta$ according to
\begin{equation}
C \leq \Delta - \Delta \log_2 \Delta
\end{equation}
for $\Delta \leq 1/e$. For small $\Delta$, the $-\Delta \log_2 \Delta$
term dominates the $\Delta$ term. We can simplify further by noting
that for all positive $\Delta$, $\sqrt{\Delta} >
-\Delta \log_2(\Delta)$. Hence, $C = O(\sqrt{\Delta})$. Thus to
achieve a given channel capacity $C$ we need
\begin{equation}
\Delta = \Omega(C^2).
\end{equation}
By \eq{deltakappa}, this implies that achieving a channel capacity $C$
requires
\begin{equation}
|1-\kappa_{\min}^2| = \Omega(C^2),
\end{equation}
where $\kappa_{\min}$ is the condition number of the nonlinear map
$M = M_1 M_0^{-1}$. This implies that one of $M_0$ or $M_1$ must have
condition number at least $\kappa = \sqrt{\kappa_{min}} = \Omega \left(
  (1-C^2)^{1/4} \right)$. This in turn implies Grover search with one
query and $O(\log_\kappa(N))$ applications of the nonlinear map via
the methods of \cite{Abrams_Lloyd}.
\end{proof}

\subsection{Channel Capacity and Total Variation Distance}
\label{apptvd}

Alice wishes to transmit a message to Bob. If she sends zero Bob
receives a sample from $p(B|0)$ and if she sends one Bob receives a
sample from $p(B|1)$. Here, $B$ is a random variable on a finite state
space $\Gamma = \{0,1\ldots,d-1\}$. The only thing we know about this
channel is that
\begin{equation}
\label{tvd}
\left| p(B|0) - p(B|1) \right| = \delta,
\end{equation}
where $| \cdot |$ denotes the total variation distance (\emph{i.e.}
half the $l_1$ distance). In this section we derive an upper bound on
the channel capacity as a function of $\delta$. Specifically, we show
that the (asymptotic) capacity $C$ obeys
\begin{equation}
C \leq \delta - \delta \log_2 \delta.
\end{equation}

Any strategy that Bob could use for decoding Alice's message
corresponds to a decomposition of $\Gamma$ as
\begin{equation}
\label{decomp}
\Gamma = \Gamma_0 \sqcup \Gamma_1
\end{equation}
where $\Gamma_0$ are the outcomes that Bob interprets as zero and
$\Gamma_1$ are the outcomes that Bob interprets as one.

From \eq{tvd} it follows that
\begin{equation}
\label{ineq}
\left| p(b \in \Gamma_0|A=0) - p(b \in \Gamma_0|A=1) \right| \leq \delta.
\end{equation}
(The defining property of total variation distance is that this holds
for any set $\Gamma_0$.)

Let $F=0$ whenever $B \in \Gamma_0$ and $F=1$ whenever $B \in
\Gamma_1$. That is, the random variable $F$ is Bob's guess as to
Alice's message. By standard Shannon theory \cite{Cover_Thomas}, the
channel capacity is the mutual information $I(F;A)$ maximized over
Alice's choice of $p(A)$.

From \eq{ineq} it follows that
\begin{equation}
\left| p(F|A=0) - p(F|A=1) \right| \leq \delta.
\end{equation}
Let $p_\alpha$ be the probability distribution
\begin{equation}
p_\alpha(F) = \alpha p(F|A=0) + (1-\alpha) p(F|A=1)
\end{equation}
for some $\alpha \in [0,1]$. From the elementary properties of total
variation distance it follows that
\begin{equation}
\left| p_\alpha(F) - p(F|A=0) \right| \leq \delta
\end{equation}
and
\begin{equation}
\left| p_\alpha(F) - p(F|A=1) \right| \leq \delta
\end{equation}
for any choice of $\alpha$. In particular, we may set $\alpha =
p(A=0)$, in which case we have
\begin{eqnarray}
\left| p(F) - p(F|A=0) \right| & \leq & \delta \label{pf1} \\
\left| p(F) - p(F|A=1) \right| & \leq & \delta. \label{pf2}
\end{eqnarray}
Next, we recall the Fannes inequality. This says that for any two
density matrices $\rho,\sigma$ on a $d$-dimensional Hilbert space
whose trace distance satisfies $T \leq \frac{1}{e}$
\begin{equation}
\left| S(\rho) - S(\sigma) \right| \leq T \log_2 d - T \log_2 T.
\end{equation}
Specializing to the special case that $\sigma$ and $\rho$ are
simultaneously diagonalizable, one obtains the following statement
about classical entropies.
\begin{cor}
\label{classfan}
Let $p$ and $q$ be two probability distributions on a state space of
size $d$. Let $T$ be the total variation distance between $p$ and
$q$. Suppose $T \leq \frac{1}{e}$. Then
\begin{equation}
| H(p) - H(q) | \leq T \log_2 d - T \log_2 T.
\end{equation}
\end{cor}

Applying corollary \ref{classfan} to \eq{pf1} and \eq{pf2}
yields\footnote{We have used $H[p]$ to denote the entropy of a
  probability distribution $p$ and $H(R)$ to denote the entropy of a
  random variable $R$.}
\begin{eqnarray}
\left| H \left[ p(F) \right] - H \left[ p(F|A=0) \right] \right| &
\leq & \delta - \delta \log_2 \delta \\
\left| H \left[ p(F) \right] - H \left[ p(F|A=1) \right] \right| &
\leq & \delta - \delta \log_2 \delta.
\end{eqnarray}
Thus,
\begin{eqnarray}
I(F;A) & = & H(F) - H(F|A) \\
       & = & H \left[ p(F) \right] - p(A=0) H \left[ p(F|A=0) \right]
                                   - p(A=1) H \left[ p(F|A=1) \right]\\
       & \leq & \delta - \delta \log_2 \delta,    
\end{eqnarray}
which completes the derivation.

\section*{Appendix B: Violations of the Born Rule}
\stepcounter{section}
\label{born}

In this appendix we consider modification of quantum mechanics in which
states evolve unitarily, but measurement statistics are not given by
the Born rule. This is loosely inspired by the ``state dependence''
resolution of the firewalls paradox, put forth by Papadodimas and Raju
\cite{PR14}. In this theory, the measurement operators $O$ which
correspond to observables are not fixed linear operators, but rather
vary depending on the state they are operating on,
i.e. $O=O(\ket{\psi})$. (In general such dependencies lead to
nonlinearities in quantum mechanics, but Papadodimas and Raju argue
these are unobservable in physically reasonable experiments.) Recently
Marolf and Polchinski \cite{MP15} have claimed that such modifications
of quantum mechanics lead to violations of the Born rule. We do not
take a position either way on Marolf and Polchinski's claim, but use
it as a starting point to investigate how violations of the Born rule
are related to superluminal signaling and computational complexity.

Here we consider violations of the Born rule of the following form:
given a state $\ket{\psi}=\sum_x \alpha_x \ket{x}$, the probability
$p_x$ of seeing outcome $x$ is given by
\begin{equation} p_x = \frac{f(\alpha_x)}{\sum_{x'} f(\alpha_{x'})}
\end{equation} for some function $f(\alpha):\mathbb{C} \rightarrow
\mathbb{R}^+$. We assume that states in the theory evolve unitarily as
in standard quantum mechanics. One could consider more general
violations of the Born rule, in which the function $f$ depends not
only on the amplitude $\alpha_x$ on $x$, but on the amplitudes on
other basis states as well. However such a generalized theory seems
impractical to work with, so we do not consider such a theory here.

We first show that, assuming a few reasonable conditions on $f$
(namely that $f$ has a reasonably behaved derivative and that
measurement statistics do not depend on the normalization of the
state), the only way to modify the Born rule is to set
$f(\alpha)=|\alpha|^{2+\delta}$ for some $\delta\neq 0$. We then show
that in theories where the Born rule is modified, superluminal
signaling is equivalent to a speedup to Grover search. More precisely,
we show that if one can send superluminal signals using states on $n$
qubits, then one can speed up Grover search on a system with $O(n)$
qubits, and vice versa. Hence one can observe superluminal signals on
reasonably sized systems if and only if one can speed up Grover search
using a reasonable number of qubits.

We are not the first authors to examine the complexity theoretic
consequences of modifications to the Born rule.  Aaronson
\cite{Aaronson_postBQP} considered such modifications, and showed that
if $\delta$ is any constant, then such modifications allow for the
solution of \textsf{\#P}-hard problems in polynomial time. Our
contribution is to show the opposite direction, namely that a
significant speedup over Grover search implies the deviation from the
Born rule $\delta$ is large, and to connect this to superluminal
signaling.

We prove our results in several steps. First, in Theorems
\ref{bornimpliesNP} and \ref{bornimpliesFTL}, we show that deviations
in the Born rule by $\delta$ allow the solution of \textsf{NP}-hard
problems and superluminal signaling using $O(1/\delta)$ qubits. As
noted previously, Theorem \ref{bornimpliesNP} follows from the work of
Aaronson \cite{Aaronson_postBQP}, but we include a proof for
completeness.

In Theorem \ref{signalingimpliesborn} we show that, assuming one has a
superluminal signaling protocol using a shared state on $m$ qubits,
the deviation from the Born rule $\delta$ must be $ \geq
\Omega(1/m)$. Likewise in Theorem \ref{searchimpliesdelta} we show
that if one can achieve a constant factor super-Grover speedup using
$m$ qubits, that we must have $\delta \geq \Omega(1/m)$ as
well. Combining these with Theorems \ref{bornimpliesNP} and
\ref{bornimpliesFTL} shows that a super-Grover speedup on $m$ qubits
implies superluminal signaling protocols with $O(m)$ qubits and vice
versa. Supplementary Figure 1 explains the relationship between
these theorems below.

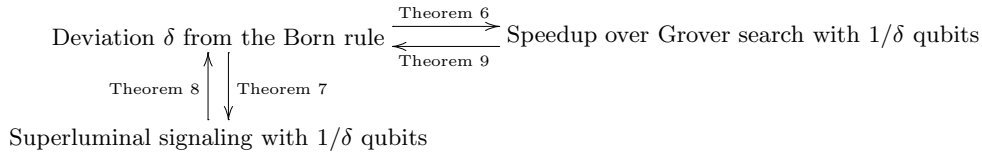
\begin{figure}[h] \label{BornTheorems} 
$\xymatrix{
{\text{Deviation } \delta \text{ from the Born rule}} \ar@<1ex>[d]^{\text{Theorem } \ref{bornimpliesFTL}} \ar@<1ex>[r]^{\text{Theorem \ref{bornimpliesNP}} \hspace{3em} } & {\text{Speedup over Grover search with }1/\delta\text{ qubits}} \ar@<1ex>[l]^{\text{Theorem } \ref{searchimpliesdelta} \hspace{3em}} \\
{\text{Superluminal signaling with }1/\delta\text{ qubits}}\ar@<1ex>[u]^{\text{Theorem }\ref{signalingimpliesborn}}}$
\caption{Relationship between theorems connecting signaling and search.}
\end{figure}
In short, we find that a violation of the Born rule by $\delta$ is
equivalent to allowing a super-Grover speedup and an instantaneous
signaling protocol using $1/\delta$ qubits. Hence in theories in which
$\delta$ is only polynomially suppressed (as a function e.g. of the
number of fields $N$ in Super-Yang-Mills), then such theories allow
for superluminal signaling and violations of the Grover lower bound
with reasonable overheads. On the other hand, our results do not rule
out violations of the Born rule in which $1/\delta$ is unphysically
large.

\subsection{Power law violations are unique}

We now show that, given some reasonable assumptions about the function
$f(\alpha)$, the only possible violation of the Born rule is
given by $f(\alpha)=|\alpha|^p$. In particular we will demand the
following properties of $f$:
\begin{enumerate}
\item Well-behaved derivative: $f(\alpha)$ is continuous and
  differentiable, and $f'(\alpha)$ changes sign at most a finite
  number of times on $[0,1]$
\item Scale invariance: for any $k\in\mathbb{C}$, we have that
  $\frac{f(k\alpha)}{\sum_x f(k\alpha_x)}=\frac{f(\alpha)}{\sum_x
    f(\alpha_x)}$. In other words the calculation of the probability
  $p_x$ of seeing outcome $x$ is independent of the norm or phase of
  the input state; it only depends on the state of the projective
  Hilbert space.
\end{enumerate}

There are a number of other reasonable constraints one could impose;
for instance one could demand that the modified Born rule has to
behave well under tensor products. Suppose you have a state
$\ket{\psi}=\sum_x \alpha_x \ket{y}$ and a state $\ket{\phi}=\sum_y
\beta_y \ket{y}$. A reasonable assumption would be to impose that in
the state $\ket{\psi}\otimes \ket{\phi}$, the probability $p_{xy}$ of
measuring outcome $xy$ should be equal to $p_xp_y$, i.e. a tensor
product state is equivalent to independent copies of each state. More
formally this would state that
\begin{equation}
\frac{f(\alpha_{x}\beta_y)}{\sum_{x'y'}f(\alpha_{x'}\beta_{y'})} =
\frac{f(\alpha_{x})}{\sum_{x'}f(\alpha_{x'})}
\frac{f(\beta_{y})}{\sum_{y'}f(\beta_{y'})}.
\end{equation}
Let us call this the Tensor product property. It will turn out that
the Tensor product property is implied by the Scale invariance
property, which we will show in our proof.

We now show that the Well-behaved derivative and Scale invariance
properties imply $f(\alpha)=|\alpha|^p$ for some $p$.

\begin{theorem} Suppose that $f$ satisfies the Well-behaved derivative
  and Scale invariance properties. Then $f(\alpha)=|\alpha|^p$ for
  some $p\in\mathbb{R}$.
\end{theorem}
\begin{proof}
First, note that the functions $f(\alpha)$ and $cf(\alpha)$ give the
same measurement statistics for any scalar $c\in\mathbb{R}$. To
eliminate this redundancy in our description of $f$, we'll choose $c$
such that $f(1)=1$.

For any $\alpha\in\mathbb{C}$, consider the (non-normalized) state
$\alpha\ket{0}+\ket{1}$. By scale invariance, for any
$\beta\in\mathbb{C}$, we must have that
\begin{equation}
\frac{f(\alpha)}{f(\alpha)+f(1)}=\frac{f(\alpha\beta)}{f(\alpha\beta) + f(\beta)}
\end{equation}
which implies that
$f(\alpha)f(\beta)=f(\alpha\beta)f(1)=f(\alpha\beta)$ for all
$\alpha,\beta \in\mathbb{C}$. One can easily check that this implies
the tensor product property.

In particular this holds for any phase, so if
$\alpha=|\alpha|e^{i\theta}$, we must have that
$f(\alpha)=\hat{f}(|\alpha|)g(\theta)$ for some functions
$\hat{f}:\mathbb{R}^{\geq0}\rightarrow \mathbb{R}^+$ and
$g:[0,2\pi)\rightarrow \mathbb{R}^+$.  Note that taking $g\rightarrow
cg$ and $\hat{f}\rightarrow \hat{f}/c$ leaves $f$ invariant for any
scalar $c\in\mathbb{R}^+$. So without loss of generality, since
$f(1)=1$, we can set $\hat{f}(1)=g(0)=1$ as well by an appropriate
choice of scalar $c$. Now, for any phases $e^{i\theta}$ and
$e^{i\phi}$, we have
$f(e^{i\theta})f(e^{i\phi})=f(e^{i(\theta+\phi)})$. Since
$\hat{f}(1)=1$ this implies $g(\theta)g(\phi)=g(\theta+\phi)$,
i.e. $g$ must be a real one-dimensional representation of $U(1)$. The
only such representation is $g=1$, hence $f(\alpha)=f(|\alpha|)$.

Now we will show that $f(x)=x^p$ for some $p$. Consider any $0<x<1$
and $0<x'<1$ where $x\neq x'$. Since
$f(\alpha)f(\beta)=f(\alpha\beta)$, we must have that $f(x^k)=f(x)^k$
and $f(x'^k)=f(x')^k$ for any $k\in\mathbb{N}$. Let
$p=\log(f(x))/\log(x)$ and $p'=\log(f(x'))/\log(x')$. Then the above
equations imply that $f(x^k)=x^{kp}$ and $f(x'^k)=x'^{kp'}$ for all
$k\in\mathbb{N}$.

Now suppose by way of contradiction that there exist $x,x'$ such that
$p\neq p'$. Since both $x<1$ and $x'<1$, as $k\rightarrow\infty$ we
have that $f(x^{kp})\rightarrow0$ and $f(x'^{kp})\rightarrow
0$. However, the sequence of points $f(x),f(x^2),f(x^3),\ldots$
approaches zero along the curve $h(x)=x^p$ while the sequence of
points $f(x'),f(x'^2),f(x'^3),\ldots$ approaches zero along the curve
$h'(x)=x^{p'}$. This implies $f$ must oscillate infinitely many times
between the curves $h$ and $h'$, which implies $f'$ must change signs
infinitely many times by the intermediate value theorem. This
contradicts the Well-behaved derivative assumption.

Hence we have for all $0<x<1$, $f(x)=x^p$ for some $p$. Now if $x>1$,
we have $f(x)f(1/x)=f(1)=1$. Since $1/x <1$, then we have
$f(1/x)=1/x^p$, so $f(x)=x^p$ as well. Also $f(1)=1^p=1$, and by
continuity $f(0)=0$. Hence for all $x\geq0$ we must have $f(x)=x^p$
for some $p$, as claimed.
\end{proof}

\subsection{Born rule violations imply signaling and super-Grover speedup}

We first show that large violations of the Born rule imply a large
speedup to Grover search and allow for large amounts of superluminal
signaling. This was previously shown by Aaronson
\cite{Aaronson_postBQP}, but for completeness we will summarize the
proof here.

\begin{theorem}[Aaronson \cite{Aaronson_postBQP} Theorem
  6] \label{bornimpliesNP}Suppose that the Born rule is modified such
  that $f(\alpha)=|\alpha|^{2+\delta}$ where $\delta\neq 0$. Then one
  can solve \textsf{PP} problems on instances of size $n$ in time
  $O(\frac{n^2}{|\delta|})$. In particular one can search an unordered
  list of $2^n$ indices in $O(\frac{n^2}{|\delta|})$ time.
\end{theorem}
\begin{proof}
We will use the modified Born rule to simulate postselection. Suppose
one has a state $\ket{\Psi}=\sum_x (\alpha_x \ket{0} + \beta_x
\ket{1}) \ket{x}$ and wishes to simulate postselection of the first
qubit on the state $\ket{0}$.  Suppose $\delta>0$; the case $\delta<0$
follows analogously. To simulate postselection on zero, simply append
$k$ ancilla qubits in the $\ket{0}$ state. Then apply a Hadamard to
each of the ancilla qubits controlled on the first qubit being a
1. The state now evolves to
\begin{equation}
\sum_x \left( \alpha_x \ket{0} \ket{x} \ket{0}^{n/\delta} + \beta_x
  \ket{1} \ket{x} \sum_{y} 2^{-k/2} \ket{y} \right)
\end{equation}
When measuring this state in the computational basis, the probability
of measuring a $0$ on the first qubit is proportional to $\sum_x
|\alpha_x|^{2+\delta}$, while the probability of getting a $1$ on the
first qubit is proportional to $2^{-k\delta/2} \sum_x
|\beta_x|^{2+\delta}$. Hence setting $k=n/\delta$, the probability of
getting a $1$ on the first qubit is exponentially suppressed by a
factor of $2^{-n}$. This effectively postselects the first qubit to
have value $0$ as desired. The rest of the proof follows from the fact
that Aaronson's \textsf{PostBQP} algorithm to solve \textsf{PP}-hard
problems on instances of size $n$ runs in time $O(n)$ and involves
$O(n)$ postselections; hence using this algorithm to solve
\textsf{PP}-hard problems when the Born rule is violated takes time
$O(\frac{n^2}{\delta})$ as claimed.
\end{proof}

Aaronson's result also implies that large violations of the Born rule imply one can send
superluminal signals with small numbers of qubits.

\begin{theorem}[Aaronson \cite{Aaronson_postBQP}] \label{bornimpliesFTL} Suppose that the Born rule is
  modified such that $f(\alpha)=|\alpha|^{2+\delta}$ where $\delta\neq
  0$. Then one can transmit a bit superluminally in a protocol
  involving a state on $O(n/|\delta|)$ qubits which succeeds with
  probability $1-2^{-n}$. Note one can use this protocol to send
  either classical bits or quantum bits.
\end{theorem}
\begin{proof}
The proof follows almost immediately from the proof of Theorem
\ref{bornimpliesNP}. Suppose that Alice wishes to send a bit 0 or 1 to
Bob. Alice and Bob can perform the standard teleportation protocol
\cite{teleportation}, but instead of Alice sending her classical
measurement outcomes to Bob, Alice simply postselects her measurement
outcome to be 00 (i.e. no corrections are necessary to Bob's state)
using the trick in Theorem \ref{bornimpliesNP}. If Alice uses
$O(n/|\delta|)$ qubits to simulate the postselection, and then measures
her qubits, she will obtain outcome 00 with probability $1-2^{-n}$ and
the bit will be correctly transmitted as desired.
\end{proof}

\subsection{Signaling implies large power law violation}

We now show that if one can send a superluminal signal with bias
$\epsilon$ using a shared state on $n$ qubits, then the violation of
the Born rule $\delta$ must satisfy $|\delta| \geq O(\epsilon/n)$. Hence
$\delta$ and $\epsilon$ must be polynomially related. Put less
precisely, if a physically reasonable experiment can send a
superluminal signal with a nontrivial probability, then there must be
a nontrivial (and hence observable) violation of the Born rule. This
in turn, implies by Theorem \ref{bornimpliesNP} that one can solve
\textsf{NP}-hard problems with a reasonable multiplicative overhead.

\begin{theorem}\label{signalingimpliesborn}
Suppose that the Born rule is modified such that
$f(\alpha)=|\alpha|^{2+\delta}$, and suppose there is a signaling
protocol using an entangled state on $n$ qubits signaling with
probability $\epsilon$. Then $|\delta| \geq O(\frac{\epsilon}{n})$.
\end{theorem}
\begin{proof}
Consider the most general signaling protocol to send a bit of
information. Suppose that Alice and Bob share an entangled state
$\ket{\Phi}$ on $n$ qubits, $m$ of which are held by Bob and $n-m$ of
which are held by Alice. To send a zero, Alice performs some unitary
$U_0$ on her half of the state, and to send a one, Alice performs some
unitary $U_1$ on her half of the state. Bob then measures in some
fixed basis $B$. This is equivalent to the following protocol: Alice
and Bob share the state $\ket{\Psi} = U_0 \ket{\Phi}$ ahead of time,
and Alice does nothing to send a $0$, and applies $U=U_1U_0^\dagger$
to obtain $\ket{\Psi'}=U\ket{\Psi}$ send a $1$. Then Bob measures in
basis $B$. We say the protocol succeeds with probability $\epsilon$ if
the distributions seen by Bob in the case Alice is sending a $0$ vs. a
$1$ differ by $\epsilon$ in total variation distance. As shown in
section \ref{apptvd} of Appendix B, the total variation distance is polynomially
related to the capacity of the resulting classical communication
channel.

Let $\alpha_{xy}$ be the amplitude of the state $\ket{x}\ket{y}$ in
$\ket{\Psi}$, where the $\ket{x}$ is an arbitrary basis for Alice's
qubits and $\ket{y}$ are given by the basis $B$ in which Bob measures
his qubits.  Let $\alpha'_{xy}$ be the amplitude of $\ket{x}\ket{y}$
in the state $\ket{\Psi'}$, so we have $\alpha'_{xy}=\sum_{x'} U_{xx'}
\alpha_{x'y}$. In short
\begin{align}
\ket{\Psi} = \sum_{xy} \alpha_{xy}\ket{x}\ket{y} & &
\ket{\Psi'}=\sum_{xy}\alpha_{xy}' \ket{x} \ket{y}
\end{align}

Assume that $\sum_{x,y} |\alpha_{xy}|^2 =1$, i.e. the state is
normalized in the $\ell_2$ norm. Since $U$ is unitary this implies the
state $U\ket{\psi'}$ is normalized in the $\ell_2$ norm as well.

Now suppose that the protocol has an $\epsilon$ probability of
success. Let $D_0$ be the distribution on outcomes $y\in \{0,1\}^m$
when Alice is sending a zero, and $D_1$ be the distribution when Alice
is sending a $1$. Let $D_b(y)$ denote the probability of obtaining
outcome $y$ under $D_b$. Then the total variation distance between
$D_0$ and $D_1$, given by $\frac{1}{2} \sum_y|D_0(y)-D_1(y)|$, must be
at least $\epsilon$. Equivalently, there must be some event $S \subset
\{0,1\}^m$ for which
\begin{equation}
\sum_{y\in S} D_0(y)-D_1(y) \geq \epsilon
\end{equation}
and for which, for all $y\in S$, we have $D_0(y) > D_1(y)$.

Assume for the moment that $\delta>0$; an analogous proof will hold in
the case $\delta<0$. Let $N=2^n$ be the dimension of the Hilbert space
of $\ket{\Psi}$. Plugging in the probabilities $D_0(y)$ and $D_1(y)$
given by the modified Born rule, we obtain
\begin{align}
\epsilon &\leq \sum_{x\in \{0,1\}^{n-m}, y\in S} \frac{|\alpha_{xy}|^{2+\delta}}{\sum_{x'y'} |\alpha_{x'y'}|^{2+\delta}} -\frac{|\alpha'_{xy}|^{2+\delta}}{\sum_{x'y'} |\alpha'_{x'y'}|^{2+\delta}}\\
&\leq  \sum_{x\in \{0,1\}^{n-m}, y\in S} N^{\delta/2}|\alpha_{xy}|^{2+\delta} - |\alpha'_{xy}|^{2+\delta}  \label{eq:usingnorms}\\
&= \sum_{x\in \{0,1\}^{n-m}, y\in S} \left(1+\frac{\delta}{2}\log(N)\right) |\alpha_{xy}|^{2}\left(1+\delta \log|\alpha_{xy}|\right) - |\alpha'_{xy}|^{2} (1+\delta \log|\alpha'_{xy}|) + O(\delta^2) \label{eq:taylordelta} \\
&=  \sum_{x\in \{0,1\}^{n-m}, y\in S} \left( |\alpha_{xy}|^2-|\alpha'_{xy}|^2\right) + \frac{\delta}{2} \log(N) |\alpha_{xy}|^{2} + \delta  \left(|\alpha_{xy}|^2\log|\alpha_{xy}| - |\alpha'_{xy}|^2\log|\alpha'_{xy}|\right)
+ O(\delta^2) \\ 
&\leq \frac{\delta}{2} \log(N) + \frac{\delta}{2}  \sum_{x\in \{0,1\}^{n-m}, y\in S}\left(|\alpha_{xy}|^2\log|\alpha_{xy}|^2 - |\alpha'_{xy}|^2\log|\alpha'_{xy}|^2\right) + O(\delta^2)  \label{eq:deltasimplify} \\
&\leq \frac{\delta}{2} \log(N) + \frac{\delta}{2} \log(N) + O(\delta^2) = \delta n + O(\delta^2) \label{eq:deltaentropy}
\end{align}
On line (\ref{eq:usingnorms}) we used the fact that for any vector
$\ket{\phi}=\sum_y \beta_y \ket{y}$ of $\ell_2$ norm 1 over a Hilbert
space of dimension $N$, we have $N^{-\delta/2} \leq \sum_y
|\beta_y|^{2+\delta}\leq 1$ when $\delta>0$. On line
(\ref{eq:taylordelta}) we expanded to first order in $\delta$. On line
(\ref{eq:deltasimplify}) we used the fact that the first term is zero
because applying a unitary to one half of a system does not affect
measurement outcomes on the other half of the system and the second
sum is upper bounded by 1. On line (\ref{eq:deltaentropy}) we used the
fact that the sum is given by a difference of entropies of (possibly
subnormalized) probability distributions, each of which is between
zero and $\log(N)$.

Hence we have that $\delta n + O(\delta^2 )\geq \epsilon$, so to first
order in $\delta$ we must have $\delta \geq \epsilon/n$ as claimed.
\end{proof}

The following corollary follows from Theorem \ref{bornimpliesNP}, and
hence we've shown that superluminal signaling implies a super-Grover
speedup. 
\begin{cor} Suppose that the Born rule is modified such that
  $f(\alpha)=|\alpha|^{2+\delta}$, and that there is a signaling
  protocol using an entangled state on $n$ qubits which signals with
  probability $\epsilon$. Then there is an algorithm to solve
  \textsf{\#P}-hard and \textsf{NP}-hard instances of size $m$
  (e.g. \textsf{\#SAT on $m$ variables}) in time
  $O(m^2n/\epsilon)$. \label{signalingimpliessearch}
\end{cor}

\subsection{Super-Grover speedup implies signaling}

We now show that even a mild super-Grover speedup implies that
$\delta$ is large, and hence one can send superluminal signals. Our
proof uses the hybrid argument of Bennett, Bernstein, Brassard and
Vazirani \cite{BBBV} combined with the proof techniques of Theorem
\ref{signalingimpliesborn}.

\begin{theorem}Suppose that the Born rule is modified such that
  $f(\alpha)=|\alpha|^{2+\delta}$, and there is an algorithm to search
  an unordered list of $N$ items with $Q$ queries using an algorithm
  over a Hilbert space of dimension $M$. Then
\begin{equation}
\frac{1}{6} \leq \frac{2 Q}{\sqrt{N}} + |\delta| \log (M) +O(\delta^2).
\end{equation}
\label{searchimpliesdelta}
\end{theorem}
\begin{proof}
Suppose that such an algorithm exists. It must consist of a series of
unitaries and oracle calls followed by a measurement in the
computational basis.

 Let $\ket{\psi^0} = \sum_y \alpha^0_y \ket{y}$ be the state of the
 algorithm just before the final measurement when there is no marked
 item, and let $\ket{\psi^x} = \sum_y \alpha^x_y \ket{y}$ be the state
 if there is a marked item. Let $D_0$ be the distribution on $y$
 obtained by measuring $\ket{\psi^0}$ in the computational basis, and
 $D_x$ be the distribution obtained by measuring $\ket{\psi^x}$. We
 know that $\ket{\psi^0}$ and $\ket{\psi^x}$ must be distinguishable
 with 2/3 probability for every $x$. Hence we must have that the total
 variation distance between $D_0$ and $D_x$ must be at least $1/6$ for
 every $x$ (otherwise one could not decide the problem with bias
 1/6). This implies there must exist some event $S_x$ for which 
\begin{equation}
\frac{1}{6} \leq \sum_{y\in S_x} D_0(y)-D_1(y)
\end{equation}
Assume $\delta>0$; an analogous proof holds for $\delta<0$. Plugging
in the expressions for $D_0$ and $D_1$ and averaging over $x$ we
obtain 
\begin{align}
\frac{1}{6} &\leq \frac{1}{N} \sum_x \sum_{y\in S_x} \frac{|\alpha^0_y|^{2+\delta}}{\sum_{y
} |\alpha^0_{y'}|^{2+\delta}}- \frac{|\alpha^x_y|^{2+\delta}}{\sum_{y
} |\alpha^x_{y'}|^{2+\delta}}\\
&\leq  \frac{1}{N} \sum_x \sum_{y\in S_x} M^{\delta/2}|\alpha^0_y|^{2+\delta}- |\alpha^x_y|^{2+\delta} \label{eq:normchange}\\
&= \frac{1}{N} \sum_x \sum_{y\in S_x} \left(1+\frac{\delta}{2}\log(M)\right)|\alpha^0_y|^2\left(1+\delta \log|\alpha^0_y|\right) - |\alpha^x_y|^2\left(1+\delta\log|\alpha^x_y|\right) + O\left(\delta^2\right) \label{eq:deltataylor} \\
&= \frac{1}{N} \sum_x \sum_{y\in S_x} \left(|\alpha^0_y|^2 - |\alpha^x_y|^2\right)+  \frac{\delta}{2}  \log(M) |\alpha^0_y|^2 +  \frac{\delta}{2} \left(|\alpha^0_y|^2\log|\alpha^0_y|^2 - |\alpha^x_y|^2\log|\alpha^x_y|^2\right) + O\left(\delta^2\right) \\
&\leq \delta \log(M) + O(\delta^2) +\frac{1}{N} \sum_x \sum_{y\in S_x}
\left(|\alpha^0_y|^2 - |\alpha^x_y|^2\right)   \label{eq:simplify}.
\end{align}
In line (\ref{eq:normchange}) we used the fact that $M^{-\delta/2}
\leq \sum_{y} |\alpha_y|^{2+\delta} \leq 1$ for any state $\alpha_y$
normalized in the $\ell_2$ norm, in line (\ref{eq:deltataylor}) we
Taylor expanded to first order in $\delta$, and in line
(\ref{eq:simplify}) we used the fact that the sum in the second term
is upper bounded by one and the sum on third term is a difference of
entropies of (subnormalized) probability distributions which is at
most $\log(M)$. 

We next consider the final term
\begin{equation}
R \equiv \frac{1}{N} \sum_x \sum_{y \in S_x} \left( |\alpha_y^0|^2 -
|\alpha_y^x|^2 \right).
\end{equation}
Let $\hat{S}_x$ be the observable
\begin{equation}
\hat{S}_x = \sum_{y \in S_x} \ket{y} \bra{y}.
\end{equation}
Then
\begin{eqnarray}
R & = & \frac{1}{N} \sum_x \left[ \bra{\psi^0} \hat{S}_x \ket{\psi^0}
  - \bra{\psi^x} \hat{S}_x \ket{\psi^x} \right] \\
& = & \frac{1}{N} \sum_x \left[ \left( \bra{\psi^0} - \bra{\psi^x}
  \right) \hat{S}_x \ket{\psi^0} + \bra{\psi^x} \hat{S}_x \left(
  \ket{\psi^0} - \ket{\psi^x} \right) \right] \\
& \leq & \frac{2}{N} \sum_x \left\| \ket{\psi^0} - \ket{\psi^x} \right\|,
\end{eqnarray}
where the last inequality uses the fact that $\| \hat{S}_x \| =
1$. Next we note that $\sum_x \left\| \ket{\psi^0} - \ket{\psi^x}
\right\|$ is the $\ell_1$ norm of the $N$-dimensional vector whose
$x^{\mathrm{th}}$ component is  
$\left\| \ket{\psi^0} - \ket{\psi^x}\right\|$. For any $N$-dimensional
vector $\vec{v}$, $\|v\|_1 \leq \sqrt{N} \| \vec{v} \|_2$. Thus,
\begin{equation}
\label{rfin}
R \leq \frac{2}{\sqrt{N}} \sqrt{ \sum_x \left\| \ket{\psi^0} - \ket{\psi^x}
  \right\|^2 }.
\end{equation}
As shown in \cite{BBBV},
a unitary search algorithm using $Q$ oracle queries yields
\begin{equation}
\label{maxdiff}
\sum_x \left\| \ket{\psi^0} - \ket{\psi^x} \right\|^2 \leq 4 Q^2.
\end{equation}
Together, \eq{rfin} and \eq{maxdiff} imply
\begin{equation}
\label{rineq}
R \leq \frac{2Q}{\sqrt{N}}.
\end{equation}
Now, \eq{rineq} bounds the last term in \eq{eq:simplify} yielding our
final result.
\begin{equation}
\frac{1}{6} \leq \delta \log(M) + O(\delta^2) + \frac{2Q}{\sqrt{N}}.
\end{equation}
\end{proof}

The following Corollary follows immediately from Theorem
\ref{searchimpliesdelta} and Theorem \ref{bornimpliesFTL}.

\begin{cor} \label{cor:groverconstant} Suppose that the Born rule is
  modified such that $f(\alpha)=|\alpha|^{2+\delta}$, and one can
  search a list of $N=2^n$ items using $m$ qubits and $Q$
  queries. Then to first order in $\delta$, we have 
\[
|\delta| \geq \frac{1}{m}\left(\frac{1}{6} - \frac{2Q}{\sqrt{N}}\right).
\]
  In particular, if one can search an $N$ element list with $Q\leq
  \sqrt{N}/24$ queries on a state of $m$ qubits, then $|\delta| \geq
  \frac{1}{12m}$, and hence by Theorem \ref{bornimpliesFTL} one can
  send superluminal signals with probability 2/3 using $O(m)$
  qubits. \label{searchimpliessignaling}
\end{cor}

In contrast, Grover's algorithm uses $\frac{\pi}{4} \sqrt{N}$ queries
to solve search, which is optimal \cite{Zalka}. So Corollary
\ref{cor:groverconstant} shows that if one can achieve even a modest
factor of ($6 \pi \approx 19$) speedup over Grover search using $m$
qubits, then one can send superluminal signals using $O(m)$
qubits. 

\section*{Appendix C: Cloning of Quantum States}
\stepcounter{section}
\label{cloning}

One way to modify quantum mechanics is to allow perfect copying of
quantum information, or ``cloning''. As a minimal example, we will
here introduce the ability to do perfect single-qubit cloning. As with
nonlinear dynamics, care must be taken to formulate a version of
quantum cloning that is actually well defined. It is clear that
perfect single qubit cloning should take $\ket{\psi} \mapsto
\ket{\psi} \otimes \ket{\psi}$ for any single-qubit pure state. The
nontrivial task is to define the behavior of the cloner on qubits that
are entangled. It is tempting to simply define cloning in terms of the
Schmidt decomposition of the entangled state. That is, applying the
cloner to qubit $B$ induces the map $\sum_i \lambda_i \ket{i_A}
\ket{i_B} \mapsto \sum_i \lambda_i \ket{i_A} \ket{i_B} \ket{i_B}$. However, this
prescription is ill-defined due to the non-uniqueness of Schmidt
decompositions. The two decompositions of the EPR pair given in
\eq{11} and \eq{++} provide an example of the inconsistency of the
above definition.

Instead, we define our single-qubit cloner as follows.
\begin{definition}
\label{clonedef}
Let $\rho_{AB}$ be a state on a bipartite system $AB$. Let
$\rho_B$ be the reduced density matrix of $B$. Then applying the
cloner to $B$ yields
\[
\rho_{AB} \mapsto \rho_{AB} \otimes \rho_B.
\]
\end{definition}
In particular, for pure input, we have $\ket{\psi_{AB}}\bra{\psi_{AB}}
\mapsto \ket{\psi_{AB}}\bra{\psi_{AB}}\otimes \rho_B$. Thus, this
version of cloning maps pure states to mixed states in
general. Furthermore, the clones are asymmetric. The cloner takes one
qubit as input and produces two qubits as output. The two output
qubits have identical reduced density matrices. However, one of the
output qubits retains all the entanglement that the input qubit had
with other systems, whereas the other qubit is unentangled with
anything else. By monogamy of entanglement it is impossible for both
outputs to retain the entanglement that the input qubit had.

It is worth noting that the addition of nonlinear dynamics, and
cloning in particular, breaks the equivalence between density matrices
and probabilistic ensembles of pure states. Here, we take density
matrices as the fundamental objects in terms of which our generalized
quantum mechanics is defined. 

In analyzing a model of computation involving cloning, we will treat
the cloning operation as an additional gate, with the same ``cost'' as
any other. In circuit diagrams, we denote the cloning gate as follows.
\[
\includegraphics[width=0.6in]{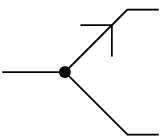}
\]
This notation reflects the asymmetric nature of our cloning gate; the
arrow indicates the output qubit that retains the entanglement of the
input qubit.

\subsection{Grover Search using Quantum Cloning}

Cloning is a nonlinear map on quantum states. As argued by Abrams and
Lloyd \cite{Abrams_Lloyd}, one can solve Grover search on a database
of size $N$ using $O(1)$ oracle queries and $O(\log N)$ applications
of $S$, for any nonlinear map $S$ from pure states to pure states,
except perhaps some pathological cases. Here, with theorem
\ref{magtheorem}, we have formalized this further, showing that this
holds as long as $S$ is differentiable. However, the cloning gate
considered here maps pure states to mixed states. Therefore, this gate
requires a separate analysis. We cannot simply invoke theorem
\ref{magtheorem}. Instead we specifically analyze the cloning gate
given above and arrive at the following result.

\begin{theorem}
\label{clonesearch}
Suppose we have access to a standard Grover bit-flip oracle, which
acts as $U_f \ket{y} \ket{z} = \ket{y} \ket{z \oplus f(y)}$ where
$f:\{0,1\}^n \to \{0,1\}$. Using one query to this oracle, followed by
a circuit using $\mathrm{poly}(n)$ conventional quantum gates and
$O(n)$ of the single-qubit cloning gates described in definition
\ref{clonedef}, one can distinguish between the cases that
$|f^{-1}(1)| = 0$ and $|f^{-1}(1)| = 1$ with high probability.
\end{theorem}

\begin{proof}
For the design of nonlinear Grover search algorithms it is helpful to
have a nonlinear map from a fixed state space to itself. To this end,
we consider circuits of the following form, which implement nonlinear
maps from the space of possible density matrices of a qubit to itself.
\[
\includegraphics[width=2in]{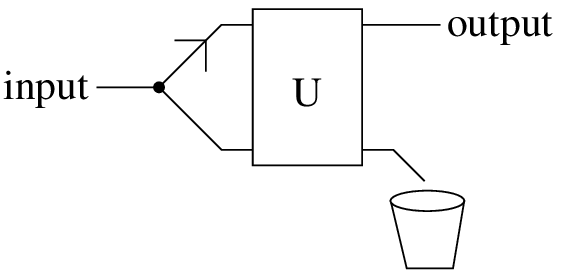}
\]
Here, one clones the input qubit, performs some unitary $U$ between
the two resulting copies, and lastly discards one of the qubits.

With a small amount of trial and error one can find a choice of $U$
which enables single-query Grover search using an analogue of the
Abrams-Lloyd algorithm. Specifically, we choose $U$ to be the
controlled-not gate. That is, let
\[
\includegraphics[width=2.2in]{cnotmap.eps}.
\]
$\mathcal{M}$ is a quadratic map on density matrices. By direct
calculation
\begin{equation}
\mathcal{M}\left( \left[ \begin{array}{cc}
r_{00} & r_{01} \\
r_{10} & r_{11} \end{array} \right] \right) 
= \left[ \begin{array}{cc}
r_{00}^2 + r_{00} r_{11} & r_{01}^2 + r_{01} r_{10} \\
r_{10}^2 + r_{10} r_{01} & r_{11}^2 + r_{11} r_{00} \end{array} \right].
\end{equation}

One can find the fixed points of $\mathcal{M}$ by solving the system
of four quadratic equations implied by $\mathcal{M}(\rho) = \rho$. The
solutions are as follows.
\begin{eqnarray}
r_{10} & = & 1 - r_{01}\textrm{,}\quad r_{11} = 1-r_{00} \label{sol1} \\
r_{00} & = & 0\textrm{,}\quad r_{10}=1-r_{01}\textrm{,}\quad r_{11} = 0 \label{sol2}\\
r_{01} & = & 1\textrm{,}\quad r_{10} = 0\textrm{,}\quad r_{11} = 1-r_{00} \label{sol3}\\
r_{00} & = & r_{01} = r_{10} = r_{11} = 0 \label{sol4}
\end{eqnarray}
The solutions \eq{sol2} and \eq{sol4} are traceless and therefore
unphysical. Solution \eq{sol3} is an arbitrary mixture of $\ket{0}$
and $\ket{1}$. That is,
\begin{equation}
\rho_r = \left[ \begin{array}{cc} r & 0 \\ 0 & 1-r \end{array} \right]
\quad \textrm{is a fixed point for all $r \in [0,1]$.}
\end{equation}
As a matrix, the solution \eq{sol1} is
\begin{equation}
\rho_{a,b} = \left[ \begin{array}{cc} a & b \\ 1-b & 1-a \end{array} \right].
\end{equation}
This is only Hermitian if $b = (1-b)^*$, which implies that $b =
\frac{1}{2} + \alpha i$ for some $\alpha \in \mathbb{R}$. However, if
$\alpha \neq 0$ then $\rho_{a,b}$ fails to be positive semidefinite,
which is unphysical. Thus, $b = \frac{1}{2}$. The eigenvalues of
$\rho_{a,1/2}$ are
\begin{equation}
\frac{1}{2} \pm \sqrt{ \frac{1 + 2a(a-1)}{2}}.
\end{equation}
Thus, unless $a = \frac{1}{2}$, the largest eigenvalue of
$\rho_{a,1/2}$ exceeds one, which is unphysical. So, the only physical
fixed point other than $\rho_r$ is
\begin{equation}
\rho_+ = \left[ \begin{array}{cc} \frac{1}{2} & \frac{1}{2} \vspace{4pt}\\
    \frac{1}{2} & \frac{1}{2} \end{array} \right] = \ket{+} \bra{+}.
\end{equation}

Numerically, one finds that $\rho_r$ is an attractive fixed point and
$\rho_+$ is a repulsive fixed point. Let
\begin{equation}
\rho_\epsilon = \left[ \begin{array}{cc} \frac{1}{2} & \frac{1}{2} -
    \epsilon \\ \frac{1}{2} - \epsilon & \frac{1}{2} \end{array}
\right] = (1-\epsilon) \ket{+}\bra{+} + \epsilon \ket{-} \bra{-}.
\end{equation}
Then,
\begin{eqnarray}
\mathcal{M}(\rho_{\epsilon}) = \rho_{2 \epsilon + O(\epsilon^2)}.
\end{eqnarray}
Consequently, $\mathcal{M}^r(\rho_{\epsilon})$ is easily
distinguishable from $\mathcal{M}^r(\rho_+) = \rho_+$ after $r = O(
\log(1/\epsilon) )$ iterations of $\mathcal{M}$.

Let $U_f$ be the standard Grover bit-flip oracle, which acts as $U_f
\ket{y}\ket{z} = \ket{y} \ket{z \oplus f(y)}$ where $f:\{0,1\}^n \to
\{0,1\}$. 
Now, consider the following circuit.
\[
\Qcircuit @C=1em @R=1em {
\lstick{\frac{\id}{2^n}} & {/} \qw & \multigate{1}{U_x} & \qw      & \qw & \\
\lstick{\ket{0}}         & \qw     & \ghost{U_x}        & \gate{H} & \qw &
}
\]
One sees that the bottom qubit emerges in the state $\rho_+$ if $f$
has no solution and emerges in the state $\rho_{\epsilon}$ with
$\epsilon = \frac{1}{2^n}$ if $f$ has one solution. By making one such
query and then applying the map $\mathcal{M}$ a total of $O(n)$ times
to the resulting state, one obtains single-qubit states in the
no-solution and one-solution cases that are easily distinguished with
high confidence using conventional quantum measurements.
\end{proof}

For simplicity, in theorem \ref{clonesearch}, we have restricted our
attention to search problems which are promised to have exactly one
solution or no solutions and our task is to determine which of these
is the case. Note that \textsf{3SAT} can be reduced to
\textsf{UNIQUESAT} in randomized polynomial time
\cite{ValiantVazirani}.  Hence solving the Grover problem in poly(n)
time when there is either exactly one solution or no solutions
suffices to solve \textsf{NP}-hard problems in randomized polynomial
time.

It is interesting to note that probability distributions also cannot
be cloned. The map $\vec{p} \mapsto \vec{p} \otimes \vec{p}$ on
vectors of probabilities is nonlinear and hence does not correspond to
any realizable stochastic process. Furthermore, one finds by a
construction similar to the above that cloning of classical
probability distributions also formally implies polynomial-time
solution to \textsf{NP}-hard problems via logarithmic-complexity single-query
Grover search. However, nonlinear maps on probabilities do not appear
to be genuinely well-defined. Suppose we have probability $p_1$ of
drawing from distribution $\vec{p}_1$ and probability $p_2$ of drawing
from distribution $\vec{p}_2$. Normally this is equivalent to
drawing from $p_1 \vec{p_1} + p_2 \vec{p_2}$. However, if we apply a
nonlinear map $\mathcal{M}$ then $\mathcal{M}(p_1 \vec{p_1} + p_2
\vec{p_2})$ is in general not equal to $p_1 \mathcal{M}(p_1) + p_2
\mathcal{M}(\vec{p_2})$. It is not clear that a well-defined
self-consistent principle can be devised for resolving such
ambiguities.

\subsection{Superluminal Signaling using Quantum Cloning}

Suppose Alice and Bob share an EPR pair $\frac{1}{\sqrt{2}} \left(
  \ket{00} + \ket{11} \right)$. If Alice wishes to transmit a zero she
does nothing. If she wishes to transmit a one she measures her qubit
in the computational basis. If Alice doesn't measure then Bob's
reduced density matrix is maximally mixed. Hence if he makes several
clones and measures them all in the computational basis he will obtain
a uniformly random string of ones and zeros. If Alice does measure
then Bob's reduced density matrix is either $\ket{0}\bra{0}$ or
$\ket{1} \bra{1}$, with equal probability. If he makes several clones
and measures them all in the computational basis he will get
$000\ldots$ or $111\ldots$, with equal probability. Thus, by making
logarithmically many clones, Bob can achieve polynomial certainty
about the bit that Alice wished to transmit. 

\section*{Appendix D: Postselection}
\stepcounter{section}
\label{postselection}

In \cite{Aaronson_postBQP} it was shown that adding the ability to
postselect a single qubit onto the state $\ket{0}$ to the quantum
circuit model yields a model of computation whose power is equal to
the classical complexity class PP. Furthermore, postselection onto
$\ket{0}$ allows perfect superluminal signaling by postselected quantum
teleportation. Here we consider a more general question: suppose we
have the ability to postselect on some arbitrary but fixed $n$-qubit
state $\ket{\psi}$. Does this still yield efficient means of solving
problems in PP and sending superluminal signals? It is clear that one
can use postselection onto $\ket{\psi}$ to simulate postselection onto
$\ket{0}$ given a quantum circuit for a unitary $U$ such that $U
\ket{00\ldots0} = \ket{\psi}$. However, for a generic $n$-qubit state
$\ket{\psi}$, no polynomial-size quantum circuit for this tasks
exists. Nevertheless, in this appendix we show that, for Haar random
(but fixed) $\ket{\psi}$, postselection onto $\ket{\psi}$ can with
high probability be used to simulate postselection onto $\ket{0}$ with
exponential precision.

We first note that the maximally entangled state of $2n$ qubits:
\begin{equation}
\ket{\Phi_{2n}} = \sum_{x \in \{0,1\}^n} \ket{x} \otimes \ket{x}
\end{equation}
can be prepared using $n$ Hadamard gates followed by $n$ CNOT
gates. Postselecting the second tensor factor of $\ket{\Phi_{2n}}$
onto $\ket{\psi}$ yields $\ket{\psi}$ on the first tensor factor. In
this manner, one may extract a copy of $\ket{\psi}$. We assume that
$\ket{\psi}$ is Haar random but fixed. That is, each time one uses the
postselection ``gate,'' one postselects onto the same state
$\ket{\psi}$. Hence, using the above procedure twice yields two copies
of $\ket{\psi}$. Applying $\sigma_x$ to one qubit of one of the copies
of $\ket{\psi}$ yields a state $\ket{\psi'} = \sigma_x \ket{\psi}$. As
shown below, the root-mean-square inner product between $\ket{\psi}$
and $\ket{\psi'}$ is of order $1/\sqrt{2^n}$. That is, they are nearly
orthogonal. Thus, one can simulate postselection onto $\ket{0}$ with
the following circuit.

\[
\Qcircuit @C=1em @R=1em {
                     & \qw     & \ctrl{1}                     & \qw \\
\lstick{\ket{\psi'}} & {/} \qw & \multigate{1}{\mathrm{SWAP}} & \qw \\
\lstick{\ket{\psi}}  & {/} \qw & \ghost{\mathrm{SWAP}}        & \gate{\ket{\psi}}
}
\]

\noindent
Here, the top qubit gets postselected onto $\ket{0}$ with fidelity
$1-O(1/\sqrt{2^n})$, the middle register is discarded, and the bottom
register is postselected onto $\ket{\psi}$, an operation we denote by
$\Qcircuit @C=1em @R=1em { & \gate{\ket{\psi}}  & \qw }$.

Lastly, we prove the claim that the root-mean-square inner
product between a Haar random $n$-qubit state $\ket{\psi}$ and
$\ket{\psi'} = \sigma_x \ket{\psi}$ is of order $1/\sqrt{2^n}$. This
mean-square inner product can be written as
\begin{eqnarray}
\bar{I} & = & \int dU | \bra{0\ldots0} U^\dag \sigma_x U
\ket{0\ldots0}|^2 \\
& = & \sum_{a,b \in \{0,1\}^n} \int dU U^\dag_{0a} U_{\bar{a} 0} U^\dag_{0b}
U_{\bar{b} 0} \label{barb}
\end{eqnarray}
where $\bar{a}$ indicates the result of flipping the first bit of $a$,
 $\bar{b}$ indicates the result of flipping the first bit of
$b$, and $0$ in the subscripts is shorthand for the bit string $0
\ldots 0$. (We arbitrarily choose the $\sigma_x$ to act on the first
qubit.)

Next we recall the following identity regarding integrals on the Haar
measure over $U(N)$. (See \cite{CS06} or appendix D of \cite{Harlow}.)
\begin{equation}
\begin{array}{rcl}
\int dU \  U_{ij} U_{kl} U^\dag_{mn} U^\dag_{op} & = & \frac{1}{N^2-1} 
\left( \delta_{in} \delta_{kp} \delta_{jm} \delta_{lo} + \delta_{ip}
  \delta_{kn} \delta_{jo} \delta_{lm} \right) \\
& & - \frac{1}{N (N^2-1)} \left( \delta_{ij} \delta_{kp} \delta_{jo}
  \delta_{lm} + \delta_{ip} \delta_{kn} \delta_{jm} \delta_{lo}
\right) \label{ident}
\end{array}
\end{equation}
Applying \eq{ident} to \eq{barb} shows that the only nonzero terms
come from $a = \bar{b}$ and consequently
\begin{eqnarray}
\bar{I} & = & \sum_{a \in \{0,1\}^n} \int dU \ U_{\bar{a} 0} U_{a0} U^\dag_{0a}
U^{\dag}_{0 \bar{a}} \\
        & = & \frac{N}{N^2-1} - \frac{1}{N^2-1}.
\end{eqnarray}
Consequently, the RMS inner product for large $N$ is
\begin{equation}
\sqrt{\bar{I}} \simeq \frac{1}{\sqrt{N}}.
\end{equation}
Recalling that $N = 2^n$ completes the argument.

\section*{Appendix E: General Nonlinearities}
\stepcounter{section}
\label{general_nonlinear}

Our discussion of final-state projection models can be thought of as
falling within a larger tradition of studying the
information-theoretic and computational complexity implications of
nonlinear quantum mechanics, as exemplified by 
\cite{Abrams_Lloyd, AaronsonNP, MW13, CY15}. A question within this subject
that has been raised multiple times \cite{Abrams_Lloyd, AaronsonNP} is
whether all nonlinearities necessarily imply that Grover search can be
solved with a single query. In this note we shed some light on this
question. However, note that the setting differs from that of section
\ref{grover} of Appendix B in that (following \cite{Abrams_Lloyd, AaronsonNP}) we 
assume the nonlinear map is the same each time, and we can apply it
polynomially many times. In section \ref{grover} of Appendix B we have included the possibility that black holes (and the nonlinear maps that they generate) are scarce and that they may differ from one another.

We first note that, for dynamics that map normalized pure states to
normalized pure states, the terms nonunitary and nonlinear are
essentially interchangeable. Let $V$ be the manifold of normalized
vectors on a complex Hilbert space $\mathcal{H}$, which could be
finite-dimensional or infinite-dimensional. Let $S:V \to V$ be a
general map, not necessarily linear or even continuous. We'll call $S$
a \emph{unitary map} if it preserves the magnitude of inner
products. That is, $|\langle S \psi|S \phi \rangle| = | \langle
\psi|\phi \rangle |$ for all $\ket{\phi},\ket{\psi} \in
\mathcal{H}$. Wigner's theorem \cite{Wigner} states that all unitary
maps are either unitary linear transformations, or antiunitary
antilinear transformations. (Antiunitary transformations are equivalent
to unitary  transformations followed by complex conjugation of all
amplitudes in some basis.) Extending quantum mechanics by allowing
antiunitary dynamics does not affect computational complexity, as can
be deduced from \cite{Aharonov}. Thus, without loss of generality, we
may ignore antiunitary maps. Hence, within the present context, if a
map is unitary then it is linear. Conversely, by linear algebra, if
map $S$ is linear, and maps $V \to V$, \emph{i.e.} is norm-preserving,
then it is also inner-product preserving, \emph{i.e.} unitary.

A standard version of the Grover problem is, for some function 
$f:\{0,1\}^n \to \{0,1\}$, to decide whether the number of solutions
to $f(y) = 1$ is zero or one, given that one of these is the case. The
search problem of finding a solution is reducible to this decision
problem with logarithmic overhead via binary search. In
\cite{Abrams_Lloyd} Abrams and Lloyd show how to solve the decision
version of Grover search using a single quantum  query to $f$ and
$O(n)$ applications of a single-qubit nonlinear map. This suffices to
solve \textsf{NP} in polynomial time. We now briefly describe their
algorithm. In contrast to section \ref{grover} of Appendix B, it is more convenient
here to assume a bit-flip oracle rather than a phase-flip oracle. That 
is, for $y \in \{0,1\}^n$ and $z \in \{0,1\}$ the oracle $O_f$
acts as
\begin{equation}
O_f \ket{y}\ket{z} = \ket{y} \ket{z \oplus f(y)}.
\end{equation}
Querying the oracle with the state 
$\frac{1}{\sqrt{2^n}} \sum_{y \in \{0,1\}^n} \ket{y} \ket{0}$  yields 
$\frac{1}{\sqrt{2^n}} \sum_{y \in  \{0,1\}^n} \ket{y} \ket{f(y)}$. 
Applying a Hadamard gate to each qubit of the first register and
measuring the first register in the computational basis yields the
outcome $00\ldots0$ with probability at least $\frac{1}{4}$. Given
that this occurs, the post-measurement state of the second register is
\begin{equation}
\ket{\psi_s} = \frac{(2^n-s) \ket{0} + s \ket{1}}{\sqrt{(2^n-s)^2+s^2}},
\end{equation}
where $s$ is the number of solutions, \emph{i.e.} $s =
|f^{-1}(1)|$. Thus, we can solve the Grover search problem by
distinguishing two exponentially-close states, namely $\ket{\psi_0}$
and $\ket{\psi_1}$. For the particular nonlinear map on the manifold
of normalized pure single-qubit states considered in
\cite{Abrams_Lloyd}, a pair of states $\epsilon$-close together
can be separated to constant distance by iterating the map
$O(\log(1/\epsilon))$ times.

We now show that any differentiable nonlinear map from pure states to
pure states on any finite-dimensional Hilbert space can achieve
this. (See theorem \ref{magtheorem}.) Let $V^{(n)}$ be the manifold of
normalized pure states on $\mathbb{C}^n$. Thus, $V^{(n)}$ is a $2n-1$
dimensional real closed compact manifold. For points $a,b$ on
$V^{(n)}$ let $|a-b|$ denote their distance. (Our choice of distance
metric is not important to the argument, but for concreteness, we
could choose the angle between quantum states, that is, $|a-b|=
\cos^{-1} |\langle a | b \rangle|$. That this is a metric is proven in
section 9.2.2 of \cite{Nielsen_Chuang}.)

\begin{theorem} Let $S:V^{(n)} \to V^{(n)}$ be a differentiable map,
  that is, a self-diffeomorphism of $V^{(n)}$. Let $r = \max_{a,b \in
    V^{(n)}} \frac{|S(a)-S(b)|}{|a-b|}$. Then there exists some
  sufficiently short geodesic $l$ in $V^{(n)}$ such that for all $x,y
  \in l$, $\frac{|S(x)-S(y)|}{|x-y|} \geq r$.
\label{magtheorem}
\end{theorem}

\begin{proof}
Choose two points $x,y$ on $V^{(n)}$ that maximize the ratio $r =
\frac{|S(x) - S(y)|}{|x-y|}$. By assumption, $S$ is not unitary, so
not all distances are preserved. Because $S$ is a map from $V^{(n)}$
to another manifold of equal volume (namely $V^{(n)}$ itself) it
cannot be that all distances are decreased. Thus, this maximum ratio
must be larger than one. The extent that this ratio exceeds one
quantifies the deviation from unitarity.

Now, consider the geodesic $g$ on $V^{(n)}$ from $x$ to $y$. Because
it is a geodesic, $g$ has length $|x-y|$. Now consider the image of
$g$ under the map $S$. Because $S$ is a continuous map, $S(g)$ will
also be a line segment. By the construction, the endpoints of $S(g)$
are distance $r|x-y|$ apart. Therefore, the length of $S(g)$, which we
denote $|S(g)|$, satisfies $|S(g)| \geq r|x-y|$, with equality
if $S(g)$ happens to also be a geodesic. Thus, $S$ induces a
diffeomorphism $S_g$ from the line segment $g$ to the line segment
$S(g)$, where $|S(g)|/|g| \geq r$. Because $S_g$ is a diffeomorphism
it follows that on any sufficiently small subsegment of $g$ it acts by
linearly magnifying or shrinking the subsegment and translating to
some location on $S(g)$. Because $|S(g)|/|g| \geq r$ it follows that
there exists some subsegment $l$ such that this linear magnification is
by a factor of at least $r$. (There could be some subsegments that
grow less than this or even shrink, but if so, others have to make up
for it by growing by a factor of more $|S(g)|/|g|$.)
\end{proof}

We now argue that the existence of $l$ suffices to ensure success for
the Abrams-Lloyd algorithm. Let $f$ denote the ``magnification
factor'' that $S$ induces on $l$. According to theorem
\ref{magtheorem}, $f \geq r$. We are interested in asymptotic
complexity, so the distance $\epsilon$ between $\ket{\psi_0}$ and
$\ket{\psi_1}$ is asymptotically small. Therefore, we assume $\epsilon$
is smaller than the length of $l$. So, we can append ancilla qubits
and apply a unitary transformation such that the resulting isometry
maps $\ket{\psi_0}$ and $\ket{\psi_1}$ to two points
$\ket{\phi_0^{(0)}}$ and $\ket{\phi_1^{(0)}}$ that lie on $l$. We then
apply $S$, resulting in the states $\ket{\phi_0^{(1)}}$ or
$\ket{\phi_1^{(1)}}$, which have distance $f \epsilon$. If $f
\epsilon$ is larger than the length $l$ then we terminate. Because we
have a fixed nonunitary map, the distance between our states is now a
constant (independent of $\epsilon$ and hence of the size of the
search space). If $f \epsilon$ is smaller than the length of $l$, then
we apply a unitary map that takes $\ket{\phi_0^{(1)}}$ and
$\ket{\phi_1^{(1)}}$ back onto $l$ and apply $S$ again. We then have
states $\ket{\phi_0^{(2)}}$ and $\ket{\phi_1^{(2)}}$ separated by
distance $f^2 \epsilon$. We then iterate this process until we exceed
the size $l$, which  separates the states to a constant distance and
uses $\log_f(1/\epsilon)$ of the nonunitary operations. States with
constant separation can be distinguished within standard quantum
mechanics by preparing a constant number of copies and collecting
statistics on the outcomes of ordinary projective measurements.

\section*{Appendix F: A Cautionary Note on Nonlinear Quantum Mechanics}
\stepcounter{section}
\label{cautionary}

The Horowitz-Maldecena final-state projection model, cloning of
quantum states, and the Gross-Pitaevsky equation (if interpreted as a
quantum wave equation) all involve nonlinear dynamics of the
wavefunction. In such cases, one must be very careful to ensure that
subsystem structure, which is captured by tensor product structure in
conventional quantum mechanics, is well-defined. Indeed, subsystem
structure is lost by introducing generic nonlinearities, and in
particular by the nonlinearity of the Gross-Pitaevsky equation. This
makes the question about superluminal signaling in the Gross-Pitaevsky
model ill-posed. The Horowitz-Maldecena model does have a natural
notion of subsystem structure, which is one of the features that makes
it appealing. Furthermore, the model of cloning that we formulate in Appendix D preserves subsystem structure by virtue of being
phrased in terms of reduced density matrices.

More formally, let $V$ be the manifold of normalized vectors in the
Hilbert space $\mathbb{C}^{d}$. We will model nonlinear quantum
dynamics by some map $S:V \to V$ which may not be a linear map on
$\mathbb{C}^d$. In general, specifying a map $S$ on $V$ does not
uniquely determine the action of $S$ when applied to a subsystem of a
larger Hilbert space. For example, consider the map $S_0$ on the
normalized pure states of one qubit given by
\begin{equation}
S_0 \ket{\psi} = \ket{0} \quad \forall \ket{\psi}
\end{equation}
Now, consider what happens if we apply $S_0$ to half of an EPR
pair $\ket{\Psi_{\mathrm{EPR}}}$. We can write the EPR state in two
equivalent ways 
\begin{eqnarray}
\ket{\Psi_{\mathrm{EPR}}} & = & \frac{1}{\sqrt{2}} \left( \ket{0}
  \ket{0} + \ket{1} \ket{1} \right) \label{11} \\
 & = & \frac{1}{\sqrt{2}} \left( \ket{+} \ket{+} +
   \ket{-} \ket{-} \right) \label{++}
\end{eqnarray}
where
\begin{equation}
\ket{\pm} = \frac{1}{\sqrt{2}} \left( \ket{0} \pm \ket{1} \right)
\end{equation}
Symbolically applying the rule $S_0 \ket{\psi} \mapsto \ket{0}$ to the
first tensor factor of \eq{11} yields $\ket{0} \ket{+}$, whereas
applying this rule to the first tensor factor of $\eq{++}$ yields
$\ket{0} \ket{0}$.

This example illustrates that one must specify additional information
beyond the action of a nonlinear map on a fixed Hilbert space in order
to obtain a well-defined extension to quantum theory incorporating the
notion of subsystems.

\section*{Appendix G: Open Problems}
\label{open}
\stepcounter{section}

We have shown that in several domains of modifications of
quantum mechanics, the resources required to observe
superluminal signaling or a speedup over Grover's algorithm are polynomially related. 
We extrapolate that this
relationship holds more generally, that is, in any quantum-like
theory, the Grover lower bound is derivable from the no-signaling
principle and vice-versa. 
A further hint in this direction is that, as shown in \cite{BHH11}, the limit on distinguishing non-orthogonal states in quantum mechanics is dictated by the no-signaling principle. Thus, any improvement over the Grover lower bound based on beyond-quantum state discrimination can be expected to imply some nonzero capacity for superluminal signaling.
There is a substantial literature on
generalizations of quantum mechanics which could be drawn upon to
address this question. In particular, one could consider the
generalized probabilistic theories framework of Barrett
\cite{Barrett05}, the category-theoretic framework of Abramsky and
Coecke \cite{AC08}, the Newton-Schr\"{o}dinger equation \cite{Ruffini}, quaternionic quantum mechanics \cite{Adler}, or the Papadodimas-Raju
state-dependence model of black hole dynamics \cite{PR14,MP15,H14}.
In these cases the investigation of computational and
communication properties is inseparably tied with the fundamental
questions about the physical interpretations of these
models. Possibly, such investigation could help shed light on these
fundamental questions. 

Our finding can be regarded as evidence
against the possibility of using black hole dynamics to efficiently
solve \textsf{NP}-complete problems, at least for problem instances of
reasonable size.
Note however that there are other
independent questions regarding the feasibility of computational
advantage through final-state projection and other forms of
non-unitary quantum mechanics. In particular, the issue of
fault-tolerance in modified quantum mechanics remains largely open,
although some discussion of this issue appears in \cite{Brun, Abrams_Lloyd,
  AaronsonNP}.
Also, while our results focus on the query complexity of search, in practice one also is interested
in the time complexity. Harlow and Hayden \cite{Harlow_Hayden} have argued that decoding the Hawking radiation emitted by a black hole may require exponential time on a quantum computer.
If the Harlow-Hayden argument is
correct, then exponential improvement in query complexity for search
does not imply exponential improvement in time-complexity. We
emphasize however that query complexity sets a lower bound on time
complexity, and therefore the reverse implication still holds, namely
exponential improvement in time complexity implies exponential
improvement in query complexity, which in the models we considered
implies superluminal signaling. 
Hence an operational version of the Grover lower bound can be derived from an operation version of the no-signaling principle.

\bibliography{supergrover}

\end{document}